\documentclass{article}

\usepackage{amsmath}
\usepackage{amsthm}
\usepackage{amssymb}
\usepackage{fullpage}
\usepackage{hyperref}
\usepackage{enumitem}
\usepackage{array}
\usepackage{rotating}
\usepackage{subcaption}
\usepackage{authblk}
\usepackage{booktabs}

\def \N {\mathbb{N}}
\def \R {\mathbb{R}}
\def \Gr {\mathbb{G}}
\def \ps {\Pi}
\def \vs {v_e}
\def \AA {A^*}
\def \RR {R^*}
\def \FA {\hat{A}}
\def \FR {\hat{R}}

\newtheorem{definition}{Definition}

\newtheorem{theorem}{Theorem}

\title{Strategic Evasion of Centrality Measures}

\author[1,2,*]{Marcin Waniek}
\author[2]{Jan Wo\'{z}nica}
\author[3]{Kai Zhou}
\author[4]{\\Yevgeniy Vorobeychik}
\author[1]{Talal Rahwan}
\author[2]{Tomasz Michalak}

\affil[1]{New York University Abu Dhabi}
\affil[2]{University of Warsaw}
\affil[3]{Hong Kong Polytechnic University}
\affil[4]{Washington University in St. Louis}
\affil[*]{To whom correspondence should be addressed: \texttt{mjwaniek@nyu.edu}}

\date{}

\begin{document}

\maketitle

\begin{abstract}
Among the most fundamental tools for social network analysis are centrality measures, which quantify the importance of every node in the network. This centrality analysis typically disregards the possibility that the network may have been deliberately manipulated to mislead the analysis. To solve this problem, a recent study attempted to understand how a member of a social network could rewire the connections therein to avoid being identified as a leader of that network. However, the study was based on the assumption that the network analyzer---the seeker---is oblivious to any evasion attempts by the evader. In this paper, we relax this assumption by modelling the seeker and evader as strategic players in a Bayesian Stackelberg game. In this context, we study the complexity of various optimization problems, and analyze the equilibria of the game under different assumptions, thereby drawing the first conclusions in the literature regarding which centralities the seeker should use to maximize the chances of detecting a strategic evader.
\end{abstract}

\section{Introduction}

Social network analysis tools have attracted significant attention in the literature \cite{getoor2005link,koschutzki2005centrality,fortunato2010community}. Such tools are typically used under an assumption that the members of the network are not strategic, i.e., they do not manipulate the topology of the network to their advantage. However, as argued by Michalak et al.~\cite{michalak2017strategic}, this assumption does not hold in many situations, ranging from privacy-savvy users of social media platforms \cite{luo2009facecloak}, through political activists \cite{youmans2012social}, to the members of criminal and terrorist organizations whose primary concern is to evade attention of security agencies \cite{kenney2013organisational}.

The first attempt to fill this gap in the literature was carried out by Waniek et al.~\cite{waniek2018hiding}, who considered how one could evade popular centrality measures, such as degree, closeness, and betweenness. More specifically, the authors studied how a member of the network---called the \emph{evader}---can rewire the network (by adding or removing edges) in order to optimally decrease the value of her centrality while maintaining her influence over other members of the network. The authors proved that, even without taking influence into consideration, the problem of decreasing the value of either closeness or betweenness centrality is NP-complete, while for the degree centrality the problem is in P.

Indeed, this study is the first in the literature to consider a strategic evader. Nevertheless, it has a number of limitations. Firstly, in their complexity analysis, the authors considered the problem of decreasing the \textit{value} of the evader's centrality, which is insufficient if the evader is concerned with decreasing her \textit{position} in the centrality-based ranking of all nodes, i.e., decreasing her centrality \emph{relative} to that of other nodes in the network. Secondly, the complexity analysis assumed that the evader is able to add and remove edges in the \textit{entire network}. This seems unrealistic in many settings such as social media platforms, where members are unable to view, let alone modify, any edge in the network. Finally, the authors assumed that the party using the social network analysis tools---the \textit{seeker}---is \textit{not strategic}, i.e., she is unaware of the evasion efforts made by the evader. While this assumption may hold in some settings, there are many others in which the seeker expects the evader to go to great lengths in order to mislead any analysis, as is the case with covert networks.

In this paper, we address all of the above limitations, and present the first analysis of evading centrality measures in settings where \textit{both parties act strategically}. We start by analyzing the complexity of decreasing the evader's \textit{position} in the centrality-based ranking, as opposed to decreasing the \textit{value} of the evader's centrality. More specifically, we require that the evader decreases her ranking by at least $d$ positions, and allow the evader to add or remove edges only \textit{locally}, i.e., in her immediate neighbourhood. We prove that this problem is NP-complete not only for closeness and betweenness centralities but also for degree centrality. Table~\ref{table:comparison} presents the main theoretical contributions of this paper.

We then model the interaction between the seeker and the evader as a \textit{Bayesian Stackelberg game} \cite{fudenberg1991game,paruchuri2007efficient,jain2008bayesian}, whereby the strategy set of the seeker consists of degree, closeness, betweenness, and eigenvector centralities, while the strategy set of the evader consists of all possible sets of changes in her network neighbourhood. Our extensive experimental analysis of this game draws the first conclusions in the literature regarding which centralities the seeker should use to maximize the chances of detecting a strategic evader.

\begin{table}[t]
\centering
\begin{tabular}{lccc}
\toprule
Centrality & Disguising & Hiding & Local Hiding \\
& Centrality~\cite{waniek2018hiding} & {Leader~\cite{waniek2017construction}} & (this paper) \\
\midrule
Degree & P & NP-complete & \textbf{NP-complete}\\
Closeness & NP-complete & NP-complete & \textbf{NP-complete}\\
Betweenness & NP-complete & unknown & \textbf{NP-complete} \\
\bottomrule
\end{tabular}
\caption{Comparing our complexity results to the literature.}
\label{table:comparison}
\end{table}

\section{Preliminaries}

Let $G = (V, E)\in \Gr$ denote a network, where $V$ is the set of $n$ nodes and $E \subseteq V \times V$ is the set of edges, and let $\Gr(V)$ denote the set of all possible networks whose set of nodes is $V$.
We denote by $(v,w)$ the edge between nodes $v$ and $w$.
We restrict our attention to undirected networks, and thus we do not discern between edges $(v,w)$ and $(w,v)$.
We also assume that networks do not contain self-loops, i.e., $\forall_{v \in V}(v,v) \notin E$. We denote by $N(v)$ the set of neighbours of $v$, i.e., $N(v) = \{w \in V : (v,w) \in E\}$.

A \textit{path} in  $(V,E)$ is an ordered sequence of nodes, $p=\langle v_1, \ldots, v_k\rangle$, in which every two consecutive nodes are connected by an edge in $E$.
The length of a path equals the number of edges therein. For any pair of nodes, $v,w \in V$, we denote by $\ps(v,w)$ the set of all shortest paths between these two nodes, and denote by $d(v,w)$ the \textit{distance} between the two, i.e., the length of a shortest path between them.

A \textit{centrality measure} is a function, $c \colon \Gr(V) \times V \rightarrow \R$, that expresses the importance of any given node in the network~\cite{bavelas1948mathematical}. We consider four fundamental centrality measures, namely degree, closeness, betweenness, and eigenvector.

\textit{Degree centrality}~\cite{shaw1954group} of node $v$ is proportional to its degree: $c_{degr}(G,v) = |N(v)|$.
\textit{Closeness centrality}~\cite{beauchamp1965improved} assigns the highest importance to the node with the shortest average distance to all other nodes: $c_{clos}(G,v) = \frac{1}{\sum_{w \in V}d(v,w)}$.
\textit{Betweenness centrality}~\cite{anthonisse1971rush,freeman1977set} of node $v$ is proportional to the percentage of shortest paths between every pair of other nodes that go through $v$: $c_{betw}(G,v) = \sum_{w\neq w'\neq v} \frac{|\{ p \in \ps(w,w') : v \in p \}|}{|\ps(w,w')|}.$
\textit{Eigenvector centrality}~\cite{bonacich1987power} evaluates each node based on the importance of its neighbours. Formally, $c_{eig}(G,v) = x_v$, where $x$ is the eigenvector corresponding to the largest eigenvalue of the adjacency matrix of $G$.

We consider two influence models:  \textit{independent cascade} and \textit{linear threshold}. Both models can be described in terms of spreading the ``activation'' of nodes across the network.
The process starts with an \textit{active} subset of nodes called the seed set. The activation then propagates through the network in discrete time steps, whereby nodes become influenced by their previously-activated neighbours.

Formally, let $I(t)$ denote the set of nodes that are active at round $t$, with $I(1)$ being the seed set. In the independent cascade model, an activation probability $p: V \times V \rightarrow \R$ is assigned to each pair of nodes. For every round $t > 1$ each node that became active in round $t-1$ has a single chance to activate each of her inactive neighbours $w$ with probability $p(v,w)$.
In our experiments we assume that for every pair of nodes, $v,w$, we have: $p(v,w)=0.15$. As for the linear threshold model, every node, $v$, is assigned a threshold, $t_v$, sampled from the set: $\{0, \ldots, |N(v)|\}$. Then, in every round $t > 1$, each inactive node becomes activated if $|I(t-1) \cap N(v)| \geq t_v$. In our experiments, the threshold of a node, $v$, is sampled from the set $\{1, \ldots, |N(v)|\}$ uniformly at random.
Notice that this variant is slightly different than the standard linear threshold model~\cite{kempe2003maximizing}, in which edges are assigned random weights. We use this variant to stay consistent with the previous literature on the topic~\cite{waniek2017construction,waniek2018hiding}.

In both models, the process ends when there are no new active nodes, i.e., when $I(t-1)=I(t)$. The influence of $v$ is then measured as the expected number of active nodes at the end of the process, when starting with $\{v\}$ as the seed set. Computing the exact influence requires exponential computations under both models, which is intractable even for relatively small networks. Thus, in our experiments we approximate the influence using Monte Carlo sampling, stopping the process when the improvement over the last $1,000$ iterations is smaller than $0.00001$. Note that even approximating the influence of a node becomes challenging when the number of nodes reaches thousands or more.

\section{Complexity of Local Hiding}
\label{sec:complexity}

We now formally define the main computational problem of our study, and analyze its computational complexity.

\begin{definition}[Local Hiding]
This problem is defined by a tuple $(G, \vs, b, c, \FA, \FR, d)$, where $G= (V,E)$ is a network, $\vs \in V$ is the evader, $b \in \N$ is a budget specifying the maximum number of edges that can be added or removed, $c \colon \Gr(V) \times V \rightarrow \R$ is a centrality measure, $\FA \subseteq N(\vs) \times N(\vs)$ is the set of edges allowed to be added, $\FR \subseteq \{\vs\} \times N(\vs)$ is the set of edges allowed to be removed, and $d \in \N$ is the safety margin.
The goal is to identify a set of edges to be added, $\AA \subseteq \FA$, and a set of edges to be removed, $\RR \subseteq \FR$, such that $|\AA| + |\RR| \leq b$ and the resulting network $(V, (E \cup \AA) \setminus \RR )$ contains at least $d$ nodes with centrality $c$ greater than that of the evader.
\end{definition}

As mentioned in the introduction, the two key differences between the above problem of \textit{Local Hiding} and the problem of \textit{Disguising Centrality} studied by Waniek et al.~\cite{waniek2018hiding} are as follows. Firstly, instead of seeking the optimal way of decreasing the value of the evader's centrality (which may not provide sufficient cover, especially if she is still ranked among the top nodes in the network), we want the position of the evader in the centrality-based ranking of all nodes to drop below $d$. Secondly, we assume that the evader is only capable of rewiring edges within her network neighbourhood---an assumption that holds in many realistic settings, e.g., the evader is able to disconnect herself from any of her friends, or even ask two of them to befriend one another, but is unable to connect to a complete stranger at will, or ask two strangers to befriend or unfriend one another. Notice that we do not allow to add any edges incident to the evader, as in case of most centrality measures such operation can only increase the ranking of the evader.

We also comment on the key differences between our Local Hiding problem and the problem of \textit{Hiding Leaders} studied by Waniek et al.~\cite{waniek2017construction} in the context of constructing covert networks. Firstly, the authors divide the nodes into leaders and the followers, where the changes in the network are allowed only among the followers. Secondly, they only allow edges to be \textit{added} among the followers, meaning that  no edge can be removed from the network.

\begin{theorem}
\label{thrm:degree-npcomplete}
The problem of Local Hiding is NP-complete given the degree centrality measure.
\end{theorem}

\begin{proof}
The problem is trivially in NP, since after the addition of a given set of edges $\AA$ and the removal of a given set of edges $\RR$ it is possible to compute the degree centrality of all nodes in polynomial time. Next, we prove that the problem is NP-hard.
To this end, we give a reduction from the NP-complete problem of Finding $k$-Clique, where the goal is to determine whether there exist $k$ nodes in $G$ that form a clique.
Given an instance of the problem of Finding $k$-Clique, defined by $k \in \N$ and a network $G=(V,E)$, let us construct a network, $H=(V',E')$, as follows:

\begin{itemize}
\item $V' = \{\vs\} \cup V \cup \bigcup_{v_i \in V} \bigcup_{j=1}^{|N(v_i)|} \{x_{i,j}\} \cup \bigcup_{i=1}^{k-2} \{z_i\}$,
\item $E' =  \bigcup_{v_i \in V'} \{(v_i,\vs) \} \cup \bigcup_{x_{i,j} \in V'} \{(v_i,x_{i,j})\} \cup \bigcup_{z_i \in V'} \{(z_i,\vs)\} \cup \bigcup_{(v_i,v_j) \notin E} \{(v_i,v_j)\}$.
\end{itemize}

\begin{figure}[t]
\centering
\includegraphics[width=.6\linewidth]{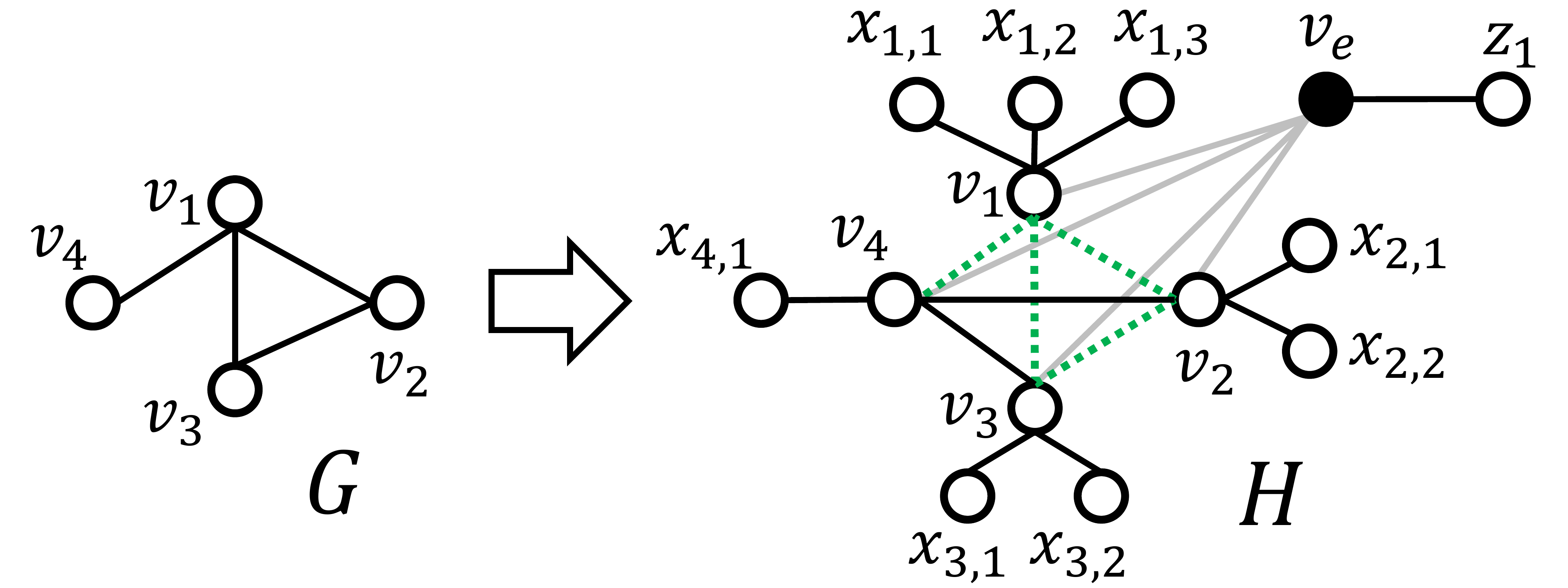}
\caption{Network used in the proof of Theorem~\ref{thrm:degree-npcomplete} for $k=3$.}
\label{fig:degree-nphard}
\end{figure}

An example of such a network $H$ is illustrated in Figure~\ref{fig:degree-nphard}. Now, consider the instance $(H,\vs,b,c,d,\FA,\FR)$ of the problem of Local Hiding where $H=(V',E')$ is the network we just constructed, $\vs$ is the evader, $b=\frac{k(k-1)}{2}$, $c$ is the degree centrality measure, $d = k$, $\FA = E$, and $\FR = \emptyset$.

From the definition of the problem we know that the edges to be added to $H$ must be chosen from $E$, i.e., from the network in the Finding $k$-Clique problem.
Out of those edges, we need to choose a subset, $\AA \subseteq E$, as a solution to the Local Hiding problem. 
In what follows, we will show that a solution to the above instance of the Local Hiding problem in $H$ corresponds to a solution to the problem of Finding $k$-Clique in $G$.

First, note that $\vs$ has the highest degree in $H$, which is $n + k -2$.
Thus, in order for $\AA$ to be a solution to the Local Hiding problem, the addition of $\AA$ to $H$ must increase the degree of at least $k$ nodes in $V$ such that each of them has a degree of at least $n + k - 1$ (note that the addition of $\AA$ only increases the degrees of nodes in $V$, since we already established that $\AA \subseteq E$).
Now since in $H$ the degree of every node $v_i$ equals $n$ (because of the way $H$ is constructed), then in order to increase the degree of $k$ such nodes to $n + k-1$, each of them must be an end of at least $k-1$ edges in $\AA$.
But since the budget in our problem instance is $\frac{k(k-1)}{2}$, then the only possible choice of $\AA$ is the one that increases the degree of exactly $k$ nodes in $V$ by exactly $k-1$.
If such a choice of $\AA$ is available, then surely those $k$ nodes form a clique in $G$, since all edges in $\AA$ are taken from $G$.
\end{proof}

\begin{theorem}
\label{thrm:closeness-npcomplete}
The problem of Local Hiding is NP-complete given the closeness centrality measure.
\end{theorem}

\begin{proof}
The problem is trivially in NP, since after the addition of a given $\AA$, and the removal of a given $\RR$, it is possible to compute the closeness centrality of all nodes in polynomial time.
Next, we prove that the problem is NP-hard.
To this end, we propose a reduction from the NP-complete $3$-Set Cover problem.
Let $U=\{u_1, \ldots, u_l\}$ be the universe, and let $S = \{S_1, \ldots, S_m\}$ be the set of subsets of the universe, where for every $S_i$ we have $|S_i|=3$.
The goal is then to determine whether there exist $k$ elements of $S$ the union of which equals $U$.
Given an instance of the $3$-Set Cover problem, let us construct a network, $G=(V,E)$, as follows:

\begin{itemize}
\item $V = \{\vs, t\} \cup \bigcup_{S_i \in S} \{S_i\} \cup \bigcup_{u_i \in U} \{u_i,w_i\} \cup \bigcup_{i=1}^{l+m-k+1} \{x_i\}$,
\item $E = \{(t,\vs)\} \cup \bigcup_{x_i \in V} \{(x_i,t)\} \cup \bigcup_{w_i \in V} \{(w_i,\vs),(w_i,u_i)\} \cup \bigcup_{S_i \in V} \{(S_i,\vs)\} \cup \bigcup_{u_j \in S_i} \{(S_i,u_j)\}$.
\end{itemize}

\begin{figure}[t]
\centering
\includegraphics[width=.5\linewidth]{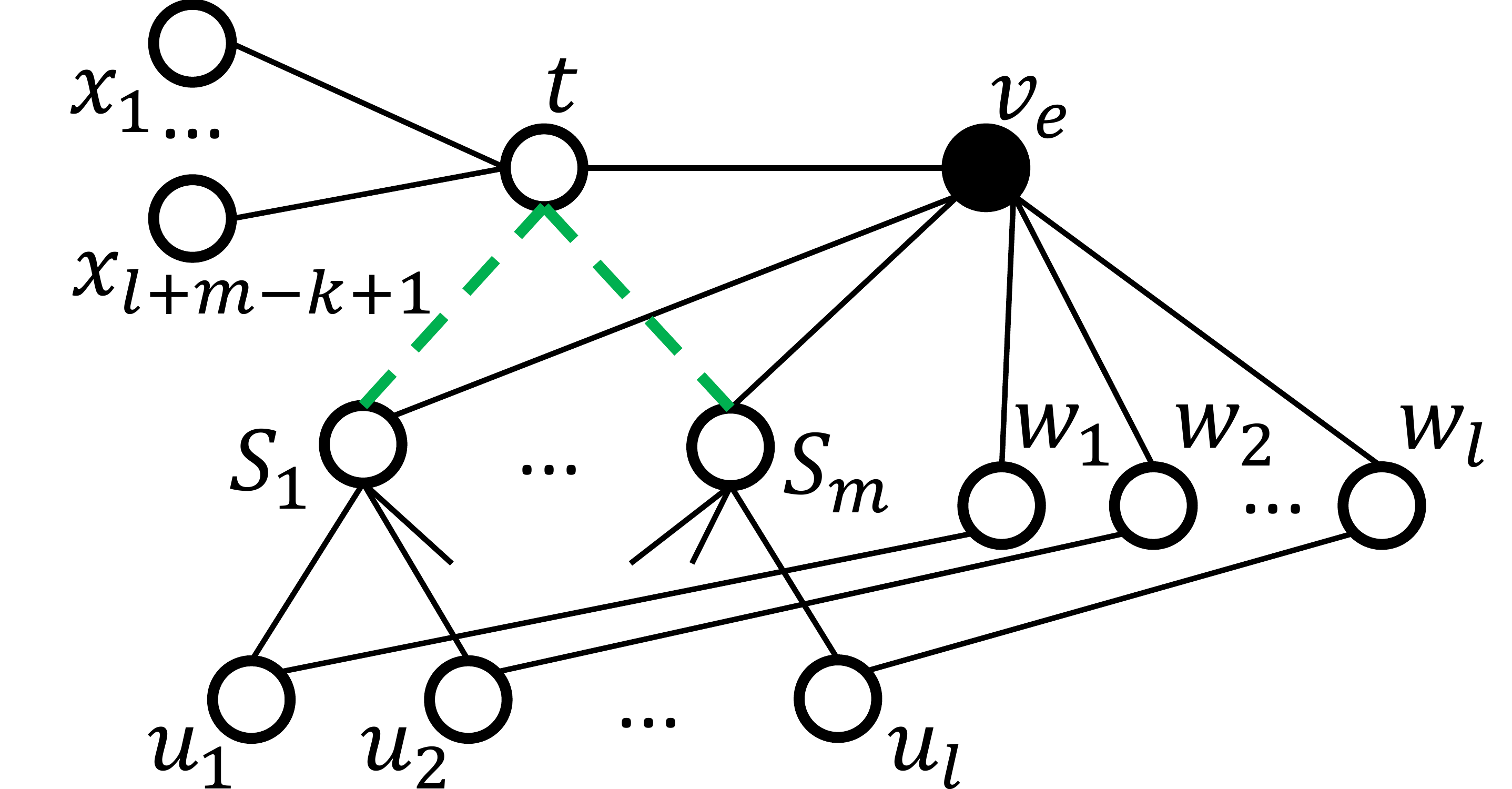}
\caption{Network used in the proof of Theorem~\ref{thrm:closeness-npcomplete}.}
\label{fig:closeness-nphard}
\end{figure}

An example of the resulting network, $G$, is illustrated in Figure~\ref{fig:closeness-nphard}. Now, consider the following instance of the problem of Local Hiding, $(G,\vs,b,c,\FA,\FR,d)$, where $G$ is the network we just constructed, $\vs$ is the evader, $b=k$ (where $k$ is the parameter of the $3$-Set Cover problem), $c$ is the closeness centrality measure, $d = 1$, $\FA = \{(t,S_i): S_i \in S\}$, and $\FR = \emptyset$.

From the definition of the problem, we see that the only edges that can be added to the graph are those between $t$ and the members of $S$.
Notice that any such choice of $\AA$ corresponds to selecting a subset of $|\AA|$ elements of $S$ in the $3$-Set Cover problem.
In what follows, we will show that a solution to the above instance of Local Hiding corresponds to a solution to the $3$-Set Cover problem.

First, we will show that for every $v \in V \setminus \{t,\vs\}$ and every $\AA \subseteq \FA$ we either have $c(G',v) < c(G',t)$ or have $c(G',v) < c(G',\vs)$, where $G' = (V,E \cup \AA)$.
To this end, let $D(G',v)$ denote the sum of distances from $v$ to all other nodes, i.e., $D(G',v) = \sum_{w \in V \setminus \{v\}} d(v,w)$.
Note that $D(G',v)= \frac{n-1}{c(G',v)}$.
We will show that the following holds:
$$
\forall_{v \in V \setminus \{t,\vs\}} \forall_{\AA \subseteq \FA} \left(D(G',v) > D(G',t) \lor D(G',v) > D(G',\vs)\right).
$$
Let $d_t$ denote $\sum_{u_i \in U} d(t,u_i) + \sum_{S_i \in S} d(t,S_i)$.
Notice also that $k \leq m$.
%Table~\ref{tab:closeness-distances} presents computation of distances between nodes in $G'$.
%In what follows, we compute $D(G',v)$ for different types of node $v$.
%Given these values we have that:
Next, we compute $D(G',v)$ for the different types of node $v$:

\begin{itemize}
\item $D(G',\vs) = 5l + 3m - 2k + 3$;
\item $D(G',t) = 3l + m - k + 2 + d_t$;
\item $D(G',x_i) = 6l + 3m - 2k + 3 + d_t > D(G',t)$;
\item $D(G',w_i) = 8l + 5m - 3k + 2 > D(G',\vs)$;
\item $D(G',u_i) \geq 9l + 4m - 3k + 2 > D(G',\vs)$ as $\sum_{S_j \in S}d(u_i,S_j) \geq m$;
\item $D(G',S_i)\! \geq\! 7l + 4m - 2k - 4\! >\! D(G',\vs)$ as $d(S_i,\vs)\! \geq\! 1$.
\end{itemize}

%\begin{table*}
%\centering
%\begin{tabular}{ l | l l l l l l | l }
%\hline
%$v$ & $d(v,\vs)$ & $d(v,t)$ & $\sum_{x_i \in X}d(v,x_i)$ & $\sum_{w_i \in W}d(v,w_i)$ & $\sum_{u_i \in U}d(v,u_i)$ & $\sum_{S_i \in S}d(v,S_i)$ & $D(G',v)$ \\
%\hline
%$\vs$ & $0$ & $1$ & $2(m+l-k+1)$ & $l$ & $2l$ & $m$ & $5l + 3m - 2k + 3$ \\ 
%$t$ & $1$ & $0$ & $m+l-k+1$ & $2l$ & $\sum_{u_i \in U} d(t,u_i)$ & $\sum_{S_i \in S} d(t,S_i)$ & $3l + m - k + 2 + d_t$ \\
%$x_i$ & $2$ & $1$ & $2(m+l-k)$ & $3l$ & $l+\sum_{u_i \in U} d(t,u_i)$ & $m+\sum_{S_i \in S} d(t,S_i)$ & $6l + 3m - 2k + 3 + d_t$ \\
%$w_i$ & $1$ & $2$ & $3(m+l-k+1)$ & $2(l-1)$ & $1+3(l-1)$ & $2m$ & $8l + 5m - 3k + 2$ \\
%$u_i$ & $2$ & $\geq 2$ & $3(m+l-k+1)$ & $1+3(l-1)$ & $3(l-1)$ & $\geq m$ & $\geq 9l + 4m - 3k + 2$ \\
%$S_i$ & $1$ & $\geq 1$ & $2(m+l-k+1)$ & $2l$ & $3+3(l-3)$ & $2(m-1)$ & $\geq 7l + 4m - 2k - 4$ \\
%\hline
%\end{tabular}
%\caption{Distances between nodes in the graph constructed for the closeness centrality proof.}
%\label{tab:closeness-distances}
%\end{table*}

Based on this, either $t$ or $\vs$ has the highest closeness centrality, therefore $\AA \subseteq \FA$ is a solution to the problem of Local Hiding if and only if $D(G',t) < D(G',\vs)$.
This is the case when $d_t < 2l + 2m - k + 1.$
Let $U_A=\{u_i \in U: \exists_{S_j \in S} u_i \in S_j \land (t,S_j) \in \AA\}$.
We have that $d_t = |\AA| + 2 (m - |\AA|) + 2 |U_A| + 3 (l-|U_A|)$ which gives us $d_t = 3l - |U_A| + 2m -|\AA|$.
Since by definition $|U_A| \leq l$ and $|\AA| \leq k$, it is possible that $d_t < 2l + 2m - k + 1$ only when $|U_A| = l$ and $|\AA| = k$, i.e., $\forall_{u_i \in U} \exists_{S_j \in S} u_i \in S_j \land (t,S_j) \in \AA$.
This solution to the problem of Local Hiding corresponds to a solution to the given instance of the $3$-Set Cover problem, which concludes the proof.
\end{proof}

\begin{theorem}
\label{thrm:betweenness-npcomplete}
The problem of Local Hiding is NP-complete given the betweenness centrality measure.
\end{theorem}

\begin{proof}
The problem is trivially in NP, since after the addition of a given set of edges $\AA$, and the removal of a given set of edges $\RR$, it is possible to compute the betweenness centrality of all nodes in polynomial time.

Next, we prove that the problem is NP-hard.
To this end, we propose a reduction from the NP-complete $3$-Set Cover problem.
Let $U=\{u_1, \ldots, u_l\}$ be the universe, and let $S = \{S_1, \ldots, S_m\}$ be the set of subsets of the universe, where for every $S_i$ we have $|S_i|=3$.
The goal is then to determine whether there exist $k$ elements of $S$ the union of which equals $U$.
Given an instance of the $3$-Set Cover problem, let us construct a network $G=(V,E)$ as follows:

\begin{itemize}
\item $V = \{\vs,t,w_1,w_2 \} \cup S \cup U \cup \bigcup_{i=1}^{\alpha} \{x_i\} \cup \bigcup_{i=1}^{\beta} \{y_i\}$, where $\alpha = m^2l(m+l+2)$ and $\beta = m^2l(k+l+2)$,
\item $E = \{(t,\vs),(w_1,w_2)\} \cup \bigcup_{x_i \in V} \{(x_i,t)\} \cup \bigcup_{y_i \in V} \{(y_i,\vs)\} \cup \bigcup_{S_i \in V} \{(S_i,\vs),(S_i,w_1)\} \cup \bigcup_{u_j \in S_i}\{(S_i,u_j)\} \cup \bigcup_{u_i \in V} \{(u_i,w_2)\} \cup \bigcup_{x_i,x_j \in V} \{(x_i,x_j)\} \cup \bigcup_{y_i,y_j \in V} \{(y_i,y_j)\}$.
\end{itemize}

\begin{figure}[t]
\centering
\includegraphics[width=.5\linewidth]{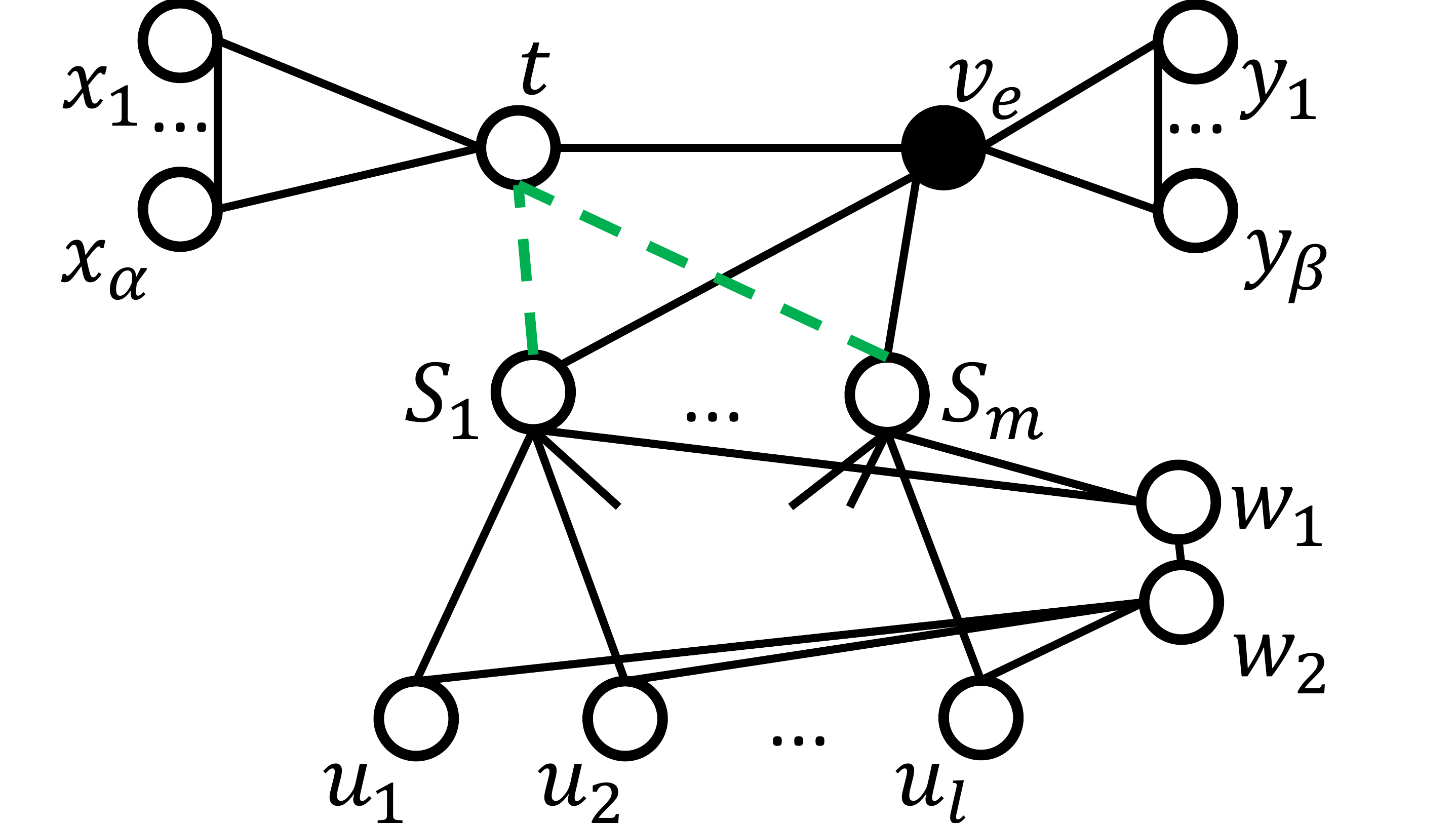}
\caption{The network used in the~proof of Theorem~\ref{thrm:betweenness-npcomplete}.}
\label{fig:betweenness-nphard}
\end{figure}

An example of the resulting network is illustrated in Figure~\ref{fig:betweenness-nphard}. Consider the instance $(G,\vs,b,c,\FA,\FR,d)$ of the problem of Local Hiding, where $G$ is the network we just constructed, $\vs$ is the evader, $b=k$ (where $k$ is the parameter of the $3$-Set Cover problem), $c$ is the betweenness centrality measure, $d = 1$, $\FA = \{(t,S_i): S_i \in S\}$, and $\FR = \emptyset$.

From the definition of the problem, one can see that the only edges that can be added to the graph are those between $t$ and the members of $S$.
Notice that any such choice of $\AA$ corresponds to selecting a subset of $|\AA|$ elements of $S$ in the $3$-Set Cover problem.
In what follows, we will show that a solution to the above instance of Local Hiding corresponds to a solution to the $3$-Set Cover problem.

First, we will show that for every node $v \in V \setminus \{t,\vs\}$ and every $\AA \subseteq \FA$ we have $c(G',v) < c(G',t)$, where $G' = (V,E \cup \AA)$.
To this end, let $B(v)$ denote the sum of percentages of shortest paths controlled by $v$ between pairs of other nodes, i.e., $B(v) = \sum_{w,w' \in V \setminus \{v\}} \frac {|\{ p \in \ps(w,w') : v \in p \}|} {|\ps(w,w')|}$.
Note that $B(v) = \frac{(n-1)(n-2)}{2} c(G',v)$
Next, we will show that the following holds:
$$
\forall_{v \in V \setminus \{t,\vs\}} \forall_{\AA \subseteq \FA} B(v) < B(t).
$$

Since $t$ controls all shortest paths between the nodes in $X$ and those in $\{\vs, w_1, w_2\} \cup Y \cup S \cup U$, we have:
$$
B(t) \geq \alpha (\beta + m + l + 3) \geq m^4 l^3 (m+l+2) + m^2l(m+l+2)^2
$$
Moreover, since $\alpha = m^2l(m+l+2)$, $\beta = m^2 l(k+l+2)$, and $k < m$, then $\alpha + \beta < 2m^2l(m+l+2)$.

For nodes other than $t$ we have:

\begin{itemize}
\item $B(x_i) = B(y_i) = 0 < B(t)$, since the nodes in $X \cup Y$ do not control any shortest paths.
\item $B(w_1) \leq (\alpha + \beta + m + 2) + \frac{m(m-1)}{2} + ml \leq  2m^2l(m+l+2) + m^2 + m + ml < (2m^2l + m)(m+l+2)< B(t)$, because $w_1$ controls some shortest paths between $w_2$ and nodes in $\{t,\vs\} \cup X \cup Y \cup S$ (there are $\alpha + \beta + m + 2$ such pairs), some shortest paths between pairs of nodes in $S$ (there are at most $\frac{m(m-1)}{2}$ such pairs), and some shortest paths between nodes in $U$ and nodes in $S$ (there are at most $ml$ such pairs).
\item $B(w_2) \leq \frac{l(l-1)}{2} + l + ml < \frac{l^2+l}{2} + ml < B(t)$, because $w_2$ controls some shortest paths between pairs of nodes in $U$ (there are at most $\frac{l(l-1)}{2}$ such pairs), some shortest paths between nodes in $U$ and $w_1$ (there are at most $l$ such pairs), and some shortest paths between nodes in $U$ and nodes in $S$ (there are at most $ml$ such pairs).
\item $B(u_i) \leq (\alpha + \beta + m + 2) + \frac{m(m-1)}{2} < B(t)$, because $u_i$ controls 
some shortest paths between $w_2$ and nodes in $\{t,\vs\} \cup X \cup Y \cup S$ (there are $\alpha + \beta + m + 2$ such pairs), and some shortest paths between pairs of nodes in $S$ (there are at most $\frac{m(m-1)}{2}$ such pairs).
\item $B(S_i) \leq 3(\alpha + \beta + l + m + 2) + l + 2(\alpha + \beta + 2) \leq 5(\alpha + \beta + l + m + 2) \leq (10 m^2l + 5)(m + l + 2) < B(t)$, because $S_i$ controls some shortest paths between the nodes in $U$ that are connected to $S_i$ and the nodes in $\{t,\vs\} \cup X \cup Y \cup S \cup U$ (there are at most $3(\alpha + \beta + l + m + 2)$ such pairs), some shortest paths between $w_1$ and the nodes in $U$ (there are at most $l$ such pairs), and some of the  shortest paths between nodes in $\{w_1,w_2\}$ and nodes in $\{t,\vs\} \cup X \cup Y$ (there are at most $2(\alpha + \beta + 2)$ such pairs).
\end{itemize}

Therefore, either $t$ or $\vs$ has the highest betweenness centrality.
Hence, $\AA \subseteq \FA$ is a solution to the problem of Local Hiding if and only if $B(t) > B(\vs)$.
We now compute the values of $B(t)$ and $B(\vs)$.
We have that:
$$
B(t) = \alpha(\beta + m + l + 3) + \sum_{\substack{S_i,S_j \in S :\\ (t,S_i) \in E \land (t,S_j) \in E}}\frac{1}{|N(S_i,S_j)|} + \sum_{S_i \in N(t)}\sum_{u_j \in U \setminus N(S_i)}\frac{|N(t,u_j)|}{|N(t,u_j)|+|N(\vs,u_j)|+1}
$$
as $t$ controls all shortest paths between every pair $(x_i,v)$ where $x_i\in X$ and $v\in V\setminus (X \cup \{t\})$ (there are $\alpha(\beta + m + l + 3)$ such pairs), one shortest path between each pair of nodes in $N(t)\cap S$, and the shortest paths between every pair $(v,w)$ where $v\in N(t)\cap S$ and $w\in U:N(t)\cap N(w)\neq \emptyset$ (other paths run through $\vs$ and nodes in $S$, or through $w_1$ and $w_2$). On the other hand, we have that:
$$
\begin{aligned}
B(\vs) = \beta(\alpha + m + l + 3) + \sum_{S_i,S_j \in S}\frac{1}{|N(S_i,S_j)|} & + \sum_{S_i \notin N(t)}(\alpha+1) + \sum_{u_i \in U : N(t,u_i) = \emptyset}(\alpha+1) \\
& + \sum_{S_i \in S}\sum_{u_j \in U \setminus N(S_i)}\frac{|N(\vs,u_j)|}{|N(t,u_j)|+|N(\vs,u_j)|+1}
\end{aligned}
$$
as $\vs$ controls all shortest paths between nodes in $Y$ and all other nodes (there are $\beta(\alpha + m + l + 3)$ such pairs), one shortest path between each pair of nodes in $S$, paths between nodes in $S$ and nodes in $U$, and all shortest paths between $\{t\} \cup X$ and nodes $\{S_i \in S : S_i \notin N(t)\} \cup \{u_i \in U : N(t,u_i) = \emptyset\}$. Thus, we have:
$$
B(\vs)-B(t) = (\beta-\alpha) (m+l+3)+\sum_{\substack{S_i,S_j \in S :\\ (t,S_i) \notin E \lor (t,S_j) \notin E}} \frac{1}{|N(S_i,S_j)|}+\Delta SU+\sum_{S_i \notin N(t)}(\alpha+1)+\sum_{u_i \in U : N(t,u_i) = \emptyset}(\alpha+1)
$$
where $0 < \Delta SU \leq ml$.

Note that $B(\vs)$ decreases with $\left|\AA\right|$ and also decreases with $\left|\{u_i \in U:\exists_{S_j \in N(t)} u_i \in S_j\}\right|$. Next, we prove that:
\begin{enumerate}[leftmargin={1cm}]
\item \label{pt1:leaders-betweenness-npcomplete} If $|\AA| = k$ and for every $u_i \in U$ there exists $S_j \in N(t)$ such that $u_i \in S_j$, then $B(\vs) < B(t)$;
\item \label{pt2:leaders-betweenness-npcomplete} If $|\AA| = k$ and there exists $u_i \in U$ such that for every $S_j \in N(t)$ we have $u_i \notin S_j$, then $B(\vs) > B(t)$.
\end{enumerate}

Regarding point~(a), we have:
$$
B(\vs)-B(t) = (\beta-\alpha) (m + l + 3) + (m - k)(\alpha + 1) +\sum_{\substack{S_i,S_j \in S :\\ (t,S_i) \notin E \lor (t,S_j) \notin E}}\frac{1}{|N(S_i,S_j)|} + \Delta SU.
$$
Now since $|\{S_i,S_j \in S : (t,S_i) \notin E \lor (t,S_j) \notin E\}|=\frac{m(m-1)-k(k-1)}{2}=\frac{(m-k)(m+k-1)}{2}$, and $|N(S_i,S_j)| \geq 2$, then we have:
$$
B(\vs)-B(t) \leq (\beta - \alpha) (m + l + 3) +(m - k)\left(\alpha + 1 + \frac{\Delta SU}{m-k} + \frac{m+k-1}{4}\right).
$$
By substituting the values of $\alpha$ and $\beta$, and observing that $\Delta SU < ml$ and $k < m$, we get:
$$
B(\vs) - B(t) < m^2l(k-m) (m+l+3) +(m-k)(m^2l(m+l+2)+1+ml+2m-1),
$$
which gives us:
$$
B(\vs) - B(t) < (k-m)m^2l + (m-k)(ml+2m) = (k-m)m(ml-l-2) < 0.
$$
Hence, if $|\AA| = k$ and for every $u_i \in U$ there exists $S_j \in N(t)$ such that $u_i \in S_j$, then $B(\vs) < B(t)$.

Regarding point~(b), since there exists $u_i \in U$ such that for every $S_j \in N(t)$ we have $u_i \notin S_j$, then:
$$
B(\vs)-B(t) \geq (\beta-\alpha) (m+l+3)+(m-k)(\alpha+1)+(\alpha+1) + \sum_{\substack{S_i,S_j \in S :\\ (t,S_i) \notin E \lor (t,S_j) \notin E}}\frac{1}{|N(S_i,S_j)|} + \Delta SU.
$$
Since $\sum_{\substack{S_i,S_j \in S :\\ (t,S_i) \notin E \lor (t,S_j) \notin E}}\frac{1}{|N(S_i,S_j)|} > 0$ and $\Delta SU > 0$, then we have:
$$
B(\vs) - B(t) > (\beta-\alpha) (m + l + 3)+(m-k+1)(\alpha+1).
$$
By substituting the values of $\alpha$ and $\beta$ we get:
$$
B(\vs)-B(t)>m^2l(k-m)(m + l + 3)+(m-k+1) (m^2l(m+l+2)+1)
$$
which gives us:
$$
B(\vs)-B(t)>m^2l(k-m)+m^2l(m+l+2)=m^2l(k+l+2) > 0
$$
Hence, if $|\AA| = k$ and there exists $u_i \in U$ such that for every $S_j \in N(t)$ we have $u_i \notin S_j$, then $B(\vs) > B(t)$.

Thus, the solution to the problem of Local Hiding corresponds to a solution to the given instance of the $3$-Set Cover problem, which concludes the proof.
\end{proof}

\section{The Seeker-Evader Game}
\label{sec:seeker-evader-game}

\textbf{Player strategies:}
We model the problem of strategically hiding in a network as a game between two players: the \emph{evader} and the \emph{seeker}. In particular, the seeker analyzes the network using a set of strategies, $T_s$, consisting of the fundamental centrality measures: degree, closeness, betweenness, and eigenvector. On the other hand, the goal of the evader is to decrease her position in the centrality-based ranking of all nodes, while maintaining her influence within the network (notice that the theoretical problems presented in Section~\ref{sec:complexity} are focused on providing safety to the evader by lowering her ranking position, while here we additionally allow the evader to take into consideration her influence in the network). To this end, she utilizes a set of strategies, $T_e$, consisting of combinations of edge modifications in her neighbourhood, with the maximum number of permitted modifications being specified by a budget, $b$.

In our experiments, we pay particular attention to the only available evader strategy in the literature, namely ROAM (Remove One Add Many)~\cite{waniek2018hiding}. In particular, the ROAM heuristic involves two steps.
\textit{Step~1:} Remove the edge between the evader, $\vs$, and its neighbour of choice, $v_0$; \textit{Step~2:} Connect $v_0$ to $b-1$ nodes of choice, who are neighbours of $\vs$ but not of $v_0$. This simple heuristic has been shown to be rather effective in practice.

\medskip

\noindent \textbf{Utility functions:}
For any given pair of strategies, $(t_s,t_e)$, such that $t_s\in T_s$ and $t_e\in T_e$, the utility of the evader is:
$$
U_e(\phi,t_s,t_e) = \phi U^R_e(t_s,t_e) + (1-\phi)U^{I}_e(t_e)
$$
where:
%\noindent where:
\begin{itemize}
\item $U^R_e(t_s,t_e) \in \R$ is the evader's utility from the change in her rank according to the centrality measure $t_s$ chosen by the seeker, when the evader plays strategy $t_e$,
\item $U^{I}_e(t_e) \in \R$ is the evader's utility from the change in her \textit{influence} within the network when she plays strategy $t_e$,
\item $\phi \in \left\{\frac{1}{m+1},\ldots,\frac{m}{m+1}\right\}$ represents the evader's evaluation of $U^R_e(t_s,t_e)$ relative to $U^{I}_e(t_e)$, we will refer to $\phi$ as the \emph{type} of the evader, with $m$ being the number of types.
\end{itemize}

Next, we specify how $U^R_e(t_s,t_e)$ and $U^{I}_e(t_e)$ are calculated (Figure~\ref{fig:util} depicts both functions). Let $r_e(t_s,t_e)$ be the evader's ranking when she plays strategy $t_e$ and the seeker plays strategy $t_s$. Then, $U^R_e(t_s,t_e)$ is calculated as follows:
\[
U^R_e(t_s,t_e) = \frac{1}{\alpha \left(1+e^{-k(r_e(t_s,t_e)-d)}\right)}-\frac{\beta}{\alpha},
\]
where $e$ is Euler's number, $k$ is the curve steepness, $d$ is the inflection point, $\beta = \frac{1}{1+e^{-k(1-d)}}$ and $\alpha = (1-2\beta)$. This formula has the following desirable properties:

\begin{itemize}
\item The evader's utility is $0$ when ranked first, i.e., fully exposed. Formally, $U^R_e(t_s,t_e) = 0$ when $r_e(v)=1$.
\item The evader's utility increases when she becomes more hidden. Formally, $U^R_e(t_s,t_e)$ increases with $r_e(t_s,t_e)$.
\item $U^R_e(t_s,t_e)$ is \textit{convex} for $1 \leq r_e(t_s,t_e)\leq d$, meaning that the marginal gain in utility increases with ranking drop, as long as the evader does not reach position $d$.
\item $U^R_e(t_s,t_e)$ is \textit{concave} for $d \leq r_e(t_s,t_e) \leq n$, i.e., dropping beyond position $d$ produces diminishing returns to the evader.
\end{itemize}

Finally, note that $U^R_e(t_s,t_e) \rightarrow 1+\frac{\beta}{\alpha}$ when $r_e(t_s,t_e) \rightarrow n$. Having specified how $U^R_e(t_s,t_e)$ is calculated, we now move to $U^{I}_e(t_e)$. Recall that the evader's influence is measured according to either the \emph{independent cascade} model or the \emph{linear threshold} model \cite{shaw1954group,beauchamp1965improved}. Regardless of which model is used, let $\Delta_e(t_e)$ denote the relative change in the evader's influence when she plays strategy $t_e$, i.e., $\Delta_e(t_e) = (I_e(t_e)-I^0_e)/I^0_e$, where $I_e(t_e)$ is the evader's influence when she plays strategy $t_e$, and $I^0_e$ is the evader's initial influence before playing. Then, $U^I_e(t_e)$ is calculated as follows:
$$
	U^I_e(t_e) = \begin{cases}
	\Delta_e(t_e), & \mbox{if } \Delta_e(t_e) > 0 \\
    -\Delta_e(t_e)^2, & \mbox{if } \Delta_e(t_e) \le 0
	\end{cases}
$$
This formula has some desired properties. Firstly, $U^I_e(t_e)$ is concave when $\Delta_e(t_e) \le 0$, meaning that the marginal loss in utility grows with the loss in influence (this is intuitive in scenarios where the evader does not mind a negligible drop in influence in return for a better disguise, but strongly opposes a significant drop in influence). Secondly, when $\Delta_e(t_e) \geq -1$, we have $U^I_e(t_e) \geq -1$, and as $\Delta_e(t_e)$ increases, $U^I_e(t_e)$ reaches a similar order of magnitude as that of $U^R_e(t_s,t_e)$, meaning that the equilibrium is not dominated by any of those two utilities. 

\begin{figure}[t]
\center
\includegraphics[width=.3\linewidth]{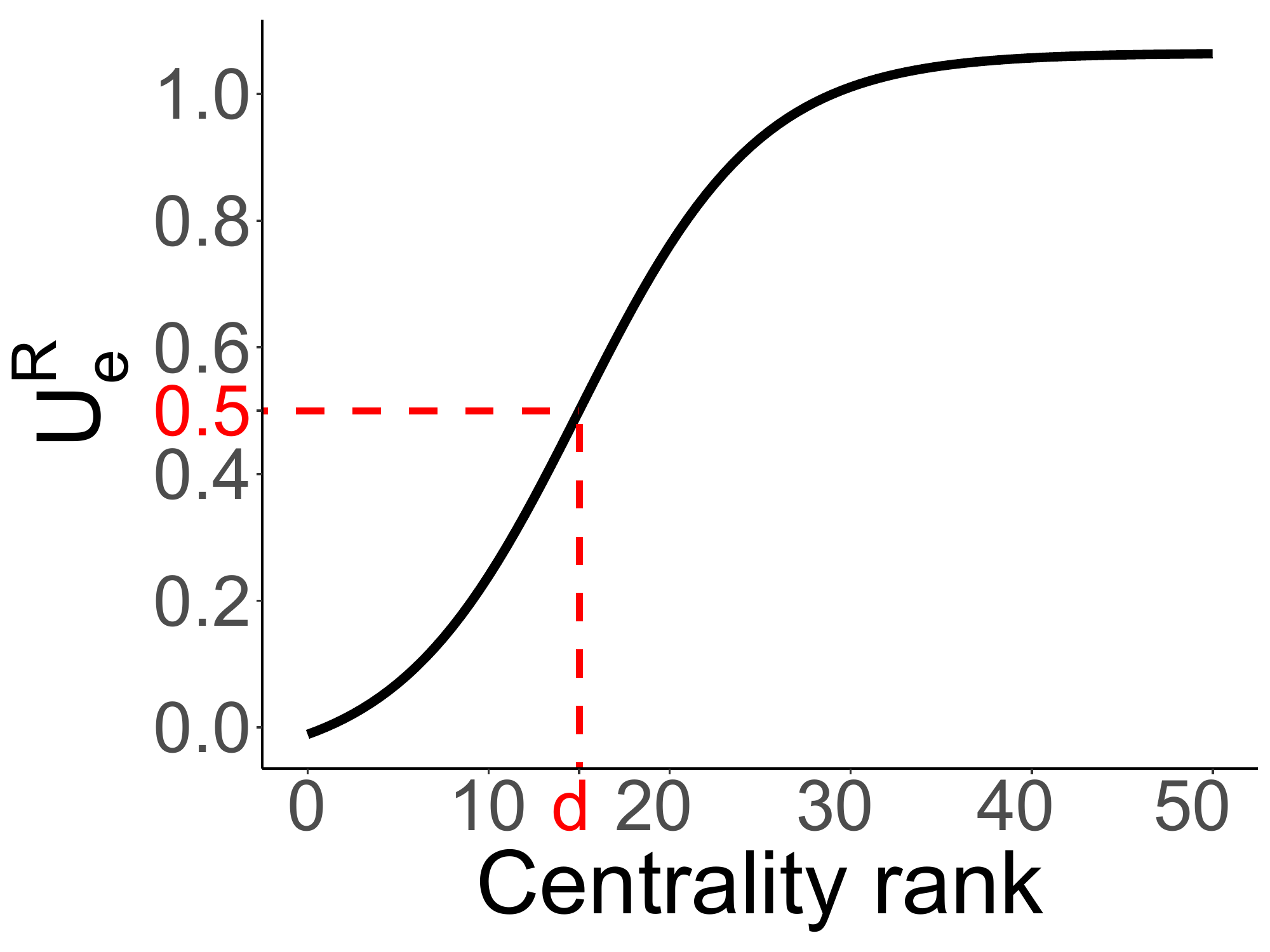}
\includegraphics[width=.3\linewidth]{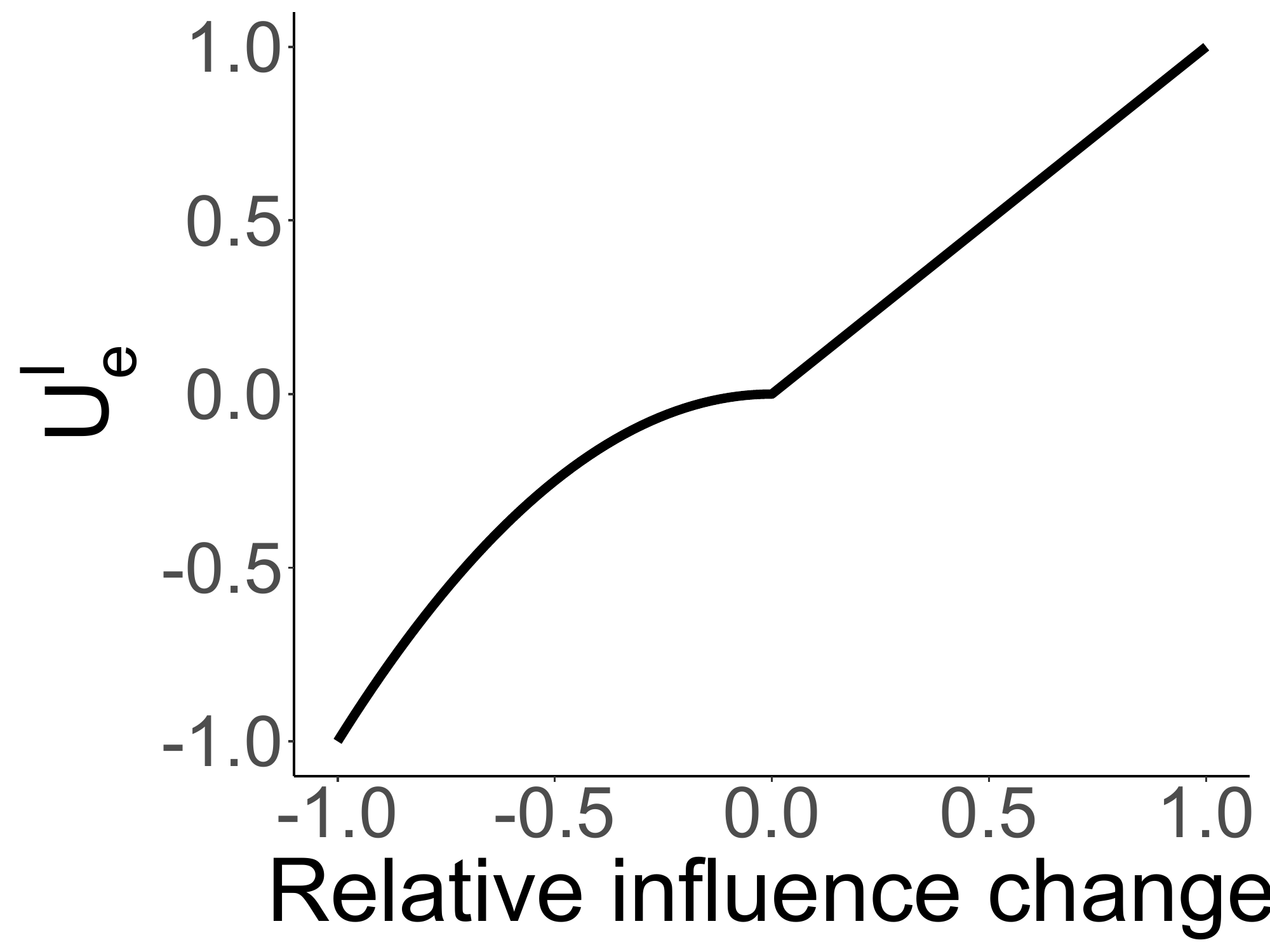}
\caption{The evader's utility functions for $d = 15$ and $k = \frac{3}{d}$.
}
\label{fig:util}
\end{figure}

Let us now turn our attention to the utility of the seeker. In our analysis we consider two different versions of the game: \emph{zero-sum game} and \emph{non-zero-sum game}. In the zero-sum version of the game we assume that the seeker is interested in minimizing the total utility of the evader, i.e., the seeker's utility is $U_s = -U_e$. In the non-zero-sum version in the game we assume that the seeker is interested solely in identifying the evader, i.e., the seeker's utility is $U_s = -U_e^R$. Notice that in the latter version of the game the seeker completely disregards any utility that the evader might gain from the change in her influence. We assume that the payoffs and the distribution of evader types are common knowledge, while the actual evader's type is private.

\medskip

\noindent\textbf{The Stackelberg game:}
Our model allows for mixed strategies. More specifically, let $p_{s}(t_s)$ be the probability that the seeker plays pure strategy $t_s\in T_s$. Moreover, let $p(\phi)$ be the probability that the evader type is $\phi$, and let $p^{\phi}_{e}(t_e)$ be the probability that an evader of type $\phi$ plays pure strategy $t_e\in T_e$. Now since the evader moves second, i.e., she knows the strategy of the seeker, then we can restrict her available strategies to only pure ones. Hence, the probability that an evader of type $\phi$ plays pure strategy $t_e\in T_{e}$ is $p_{e}^{\phi}(t_e) \in \{0,1\}$. The seeker's objective is to maximize her expected payoff. This optimization problem  can be formulated as a Mixed-Integer Quadratic problem:

\begin{alignat*}{2}
  & \text{max} && \sum_{\phi \in \Phi} \sum_{t_s\in T_{s}} \sum_{t_e\in T_{e}}  p(\phi) p^{\phi}_{e}(t_e) p_{s}(t_s) U_{s}(\phi,t_s,t_e)  \\
    	& \text{s.t.}  & \quad & \begin{aligned}[t]
        \ &\sum_{t_s\in T_{s}} p_{s}(t_s) = 1 \\
        \ &\sum_{t_e\in T_{e}} p^{\phi}_{e}(t_e) = 1  \\
    	\ &\lambda \ge \sum_{t_s\in T_{s}} p_{s}(t_s)U_e(\phi,t_s,t_e) \\
        \ &\lambda \le (1-p^{\phi}_{e}(t_e))\eta + \sum_{i\in T_{s}} p_{s}(t_s)U_e(\phi,t_s,t_e)
\end{aligned}
\end{alignat*}

The first and second constraints correspond to the probability distributions over the sets of strategies available to the players. As for $\eta\in\mathbb{R}$, it is an arbitrarily large number. This way, the third and fourth constraints ensure that, by solving the problem, we get:
$$
\lambda = \max\limits_{t_e\in T_e} \sum_{t_s\in T_{s}} p_{s}(t_s)  U_e(\phi,t_s,t_e).
$$
This is because, when $\eta$ is arbitrarily large, $(1-p^{\phi}_{e}(t_e))\eta$ reflects the fact that the evader will play the strategy that maximizes her expected payoff. Finally, in order to solve the problem efficiently, we linearize it by substituting variables: $z^{\phi}(t_s,t_e) = p^{\phi}_{e}(t_e)p_{s}(t_s)$. 
We use the linearization procedure described by Paruchuri et al.~\cite{paruchuri2008playing}.

\section{Empirical Analysis}

\begin{table}[t]
\centering
\begin{tabular}{lccc}  
\toprule
Network & Network & All & Undominated \\
& size & strategies & strategies \\
\midrule
WTC & 36 & 14190 & 60 \\
Bali & 17 & 280840 & 7 \\
Madrid & 70 & 45760 & 5 \\
Scale-Free & 30 & 61365 & 17 \\
Small-World & 30 & 902 & 36 \\
Erdos-Renyi & 30 & 4122 & 47 \\
\bottomrule
\end{tabular}
\caption{The number of possible strategies vs. the number of undominated strategies (for random networks, the number is taken as the average over $100$ such networks).}
\label{table:datasets}
\end{table}

\subsection{Network Datasets}

We now briefly describe the network datasets used in our analysis.
We consider three standard models of random networks (for each model, we generate $100$ networks consisting of $30$ nodes):

\begin{itemize}
\item \emph{Scale-free} networks, generated using the Barabasi-Albert model~\cite{barabasi1999emergence}. The number of links added with each node is $3$.
\item \emph{Small-world} networks, generated using the Watts-Strogatz model~\cite{watts1998collective}. In our experiments, the expected average degree is $10$.
\item \emph{Random graphs} generated using the Erdos-Renyi model ~\cite{erdds1959random}. In our experiments, the expected average degree is $10$.
\end{itemize}

We also analyze a number of real-life network datasets.
We consider three terrorist networks, namely:
\begin{itemize}
\item \emph{WTC}---the network of terrorists responsible for the WTC 9/11 attack~\cite{Krebs:2002a};
\item \emph{Bali}---the network of terrorists behind the 2002 Bali attack~\cite{hayes2006connecting};
\item \emph{Madrid}---the network of terrorists responsible for the 2004 Madrid train bombing~\cite{hayes2006connecting}.
\end{itemize}

Finally, we consider anonymized fragments of three social media networks, namely Facebook, Twitter and Google+~\cite{leskovec2012learning}.

The networks that we consider in our experiments are of moderate size, as for every evader’s strategy we need to compute the ranking produced by each centrality measure, which in turn requires us to compute the centrality of all nodes.

\subsection{Experimental Process}

For each network, following the work by Waniek et al.~\cite{waniek2018hiding}, the evader is chosen as the node with the smallest sum of centrality ranks (based on Degree, Closeness, Betweenness and Eigenvector); ties are broken uniformly at random. The evader type $\phi$ is sampled uniformly at random from the set $\{0.2,0.4,0.6,0.8\}$.
All results for random networks are presented as an average over $100$ samples.

While the number of pure strategies of the seeker is rather small (we assume them to be the four main centrality measures), the number of pure strategies of the evader is much larger, since every possible way of rewiring the evader's neighbourhood may be considered a unique strategy. This very quickly becomes computationally challenging even for small networks and small budgets. For instance, in the case of the WTC network, the number of the evader's strategies for budget $b=3$ is $14,190$, for $b=4$ it is $148,995$, and for $b=5$ it is $1,221,759$.

With this in mind, to study the evader's entire space of possible strategies, we focus first on a version of the game that is more computationally feasible. More specifically, we analyze the \emph{zero-sum} version of the game, where the seeker's gain equals the evader's loss. This implies that the seeker is not only interested in the evader's centrality (as in the aforementioned model), but is also interested in the evader's influence (this is implied by the fact that the evader's utility does not only depend on her centrality but also on her influence). Importantly, this version of the game can be formulated as a linear program; hence, it is much easier to solve. By analyzing the zero-sum version of the model, we aim to understand the properties of the evader's most rewarding strategies. This understanding will help us identify effective heuristics for the evader, which in turn would enable us to study the original, more computationally-challenging version of the game.

\subsection{The Zero-Sum Version}

For each network we generated the payoff matrices corresponding to budgets $3$ and $4$ and both influence measures. We were also able to consider $10\%$ of the strategies corresponding to budget $5$ (except for the WTC network, where we considered $100\%$). Our main observations regarding the strategies are threefold.

Firstly, \textit{most of the evader's strategies are strongly dominated, regardless of the evader's type}. Specifically, given different networks, Table~\ref{table:datasets} specifies the number of all strategies as well as those that are undominated. As shown, less than 1\% of strategies are undominated, and this percentage is even smaller for larger networks.

\begin{figure}[t]
\center
\includegraphics[width=0.32\linewidth]{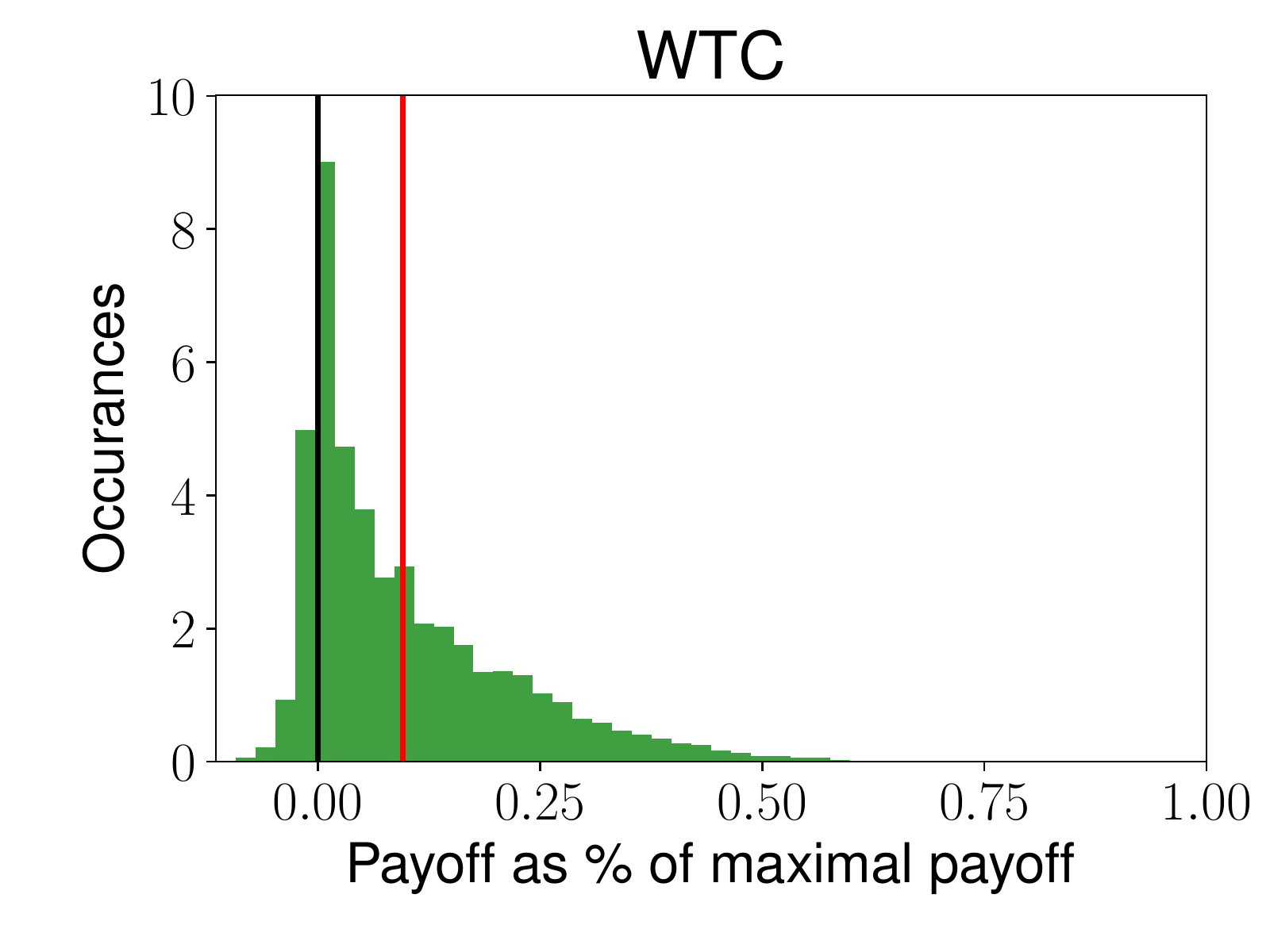}
\includegraphics[width=0.32\linewidth]{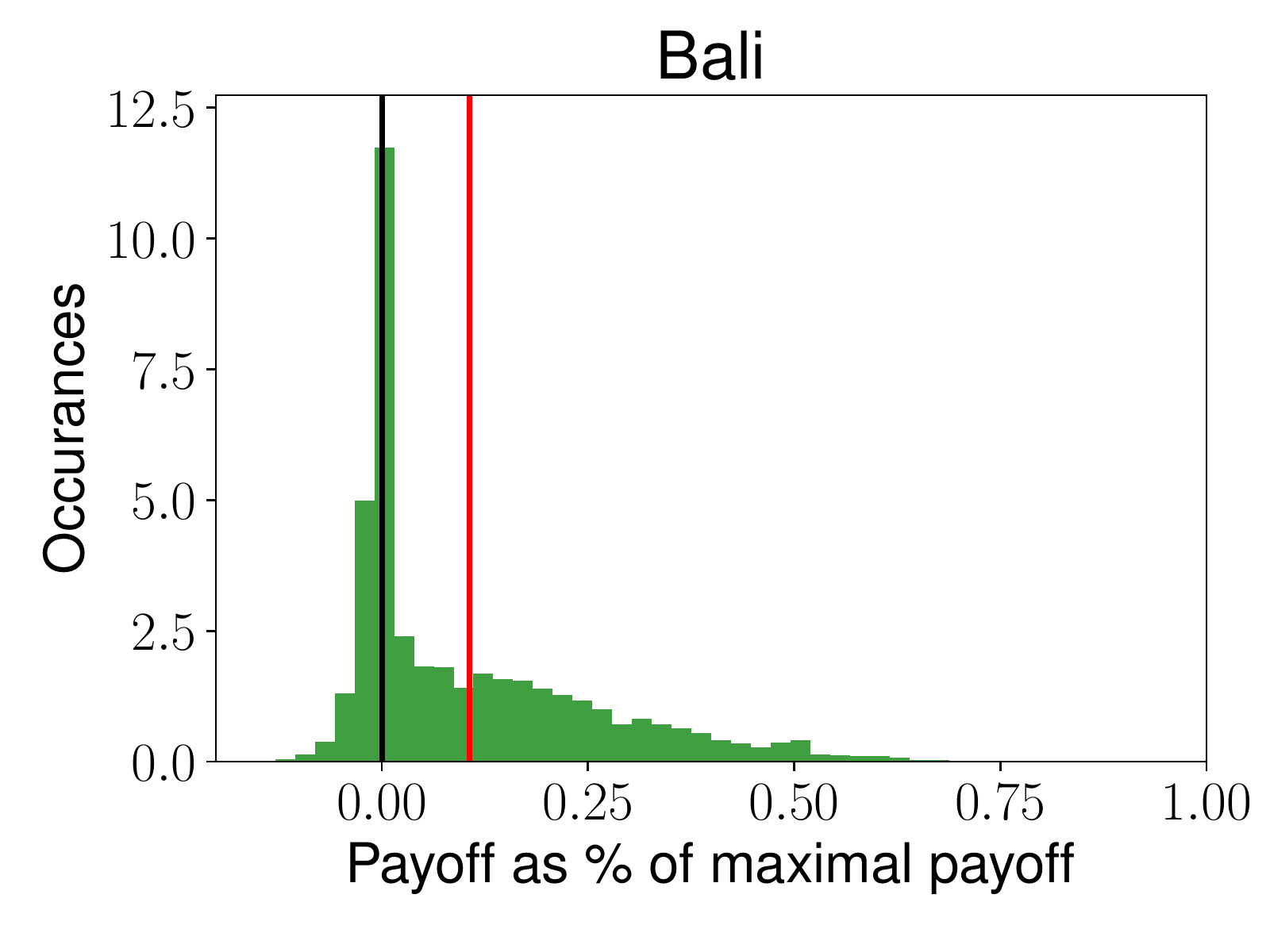}
\includegraphics[width=0.32\linewidth]{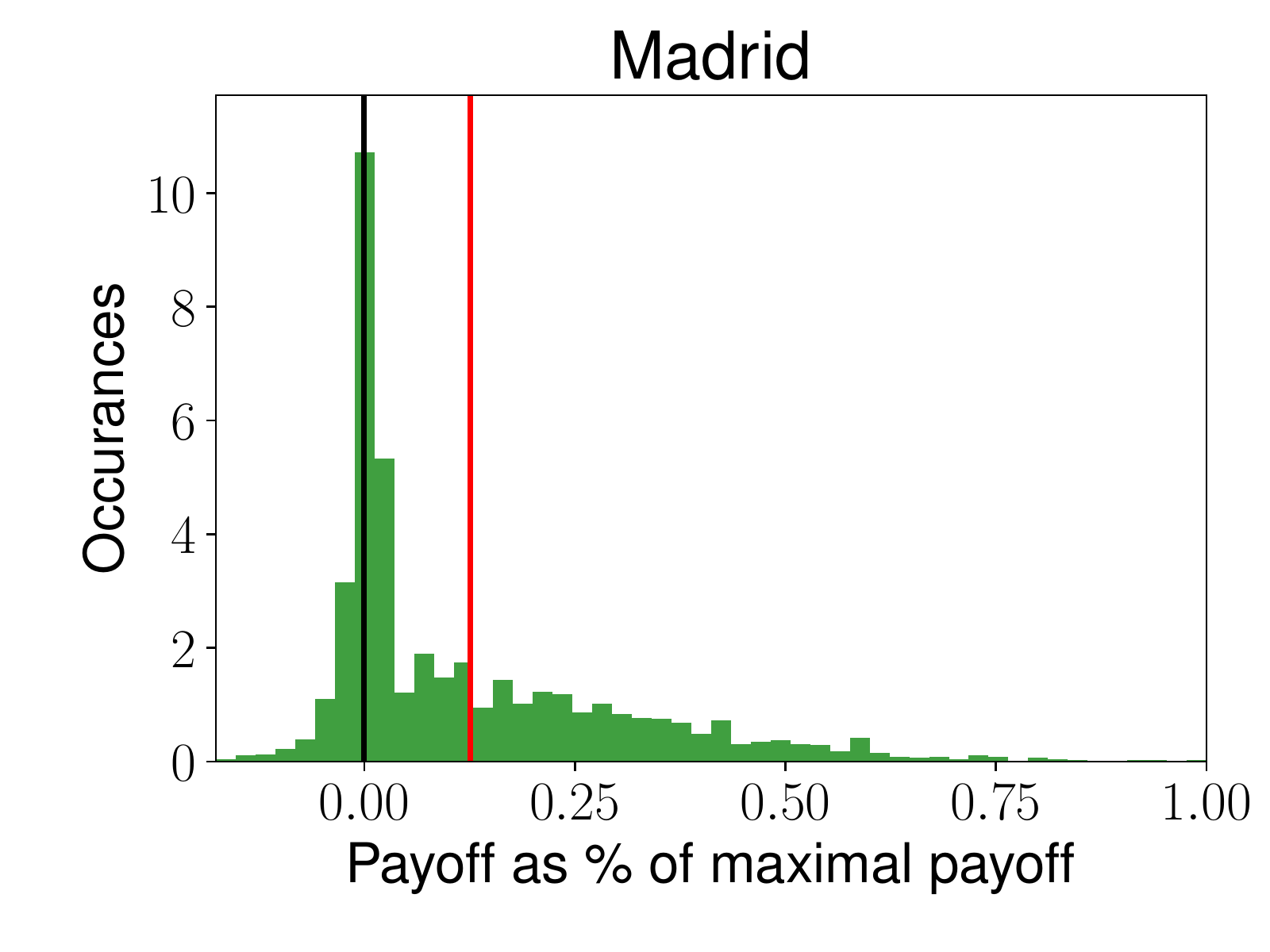}
\includegraphics[width=0.32\linewidth]{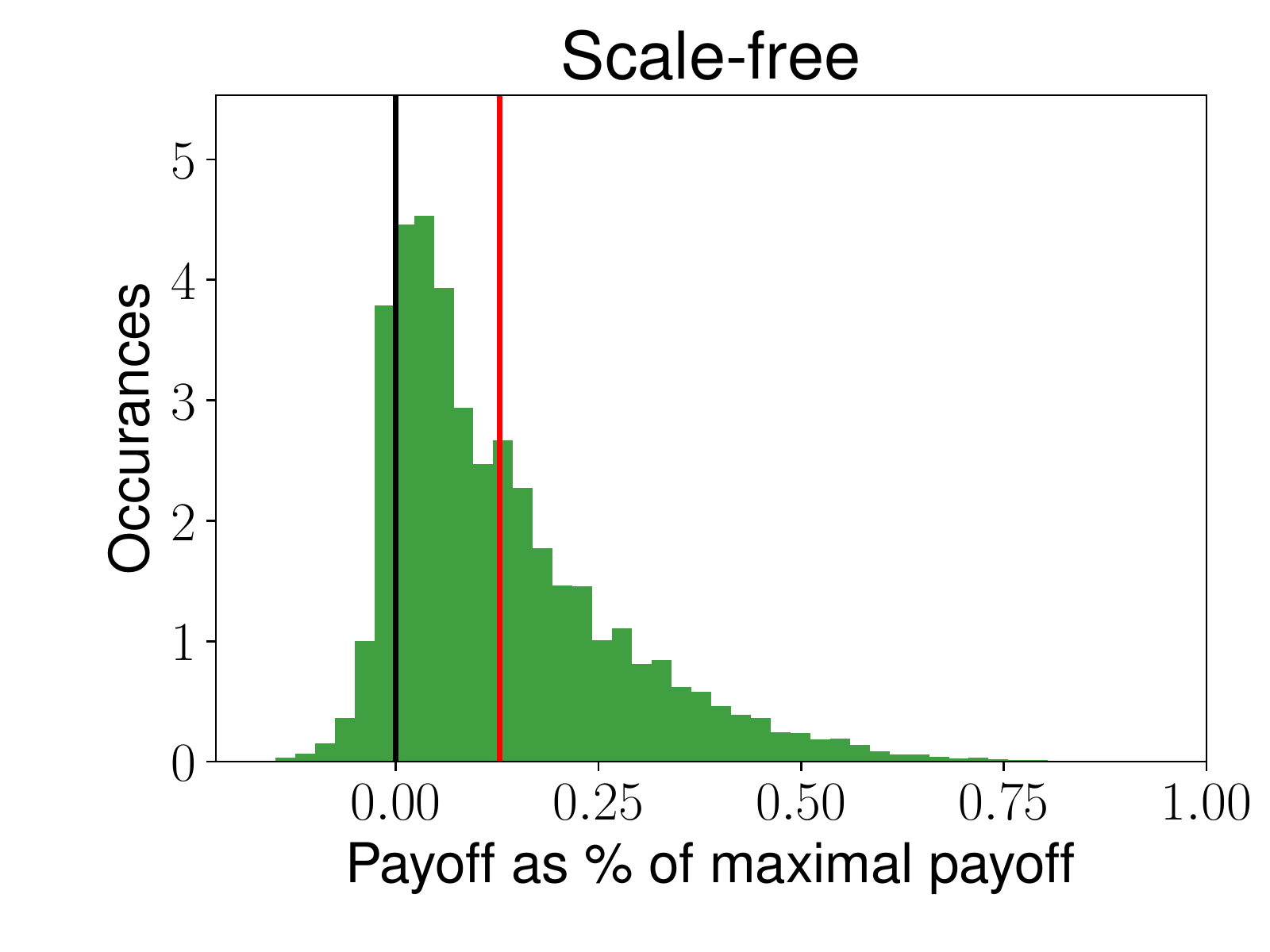}
\includegraphics[width=0.32\linewidth]{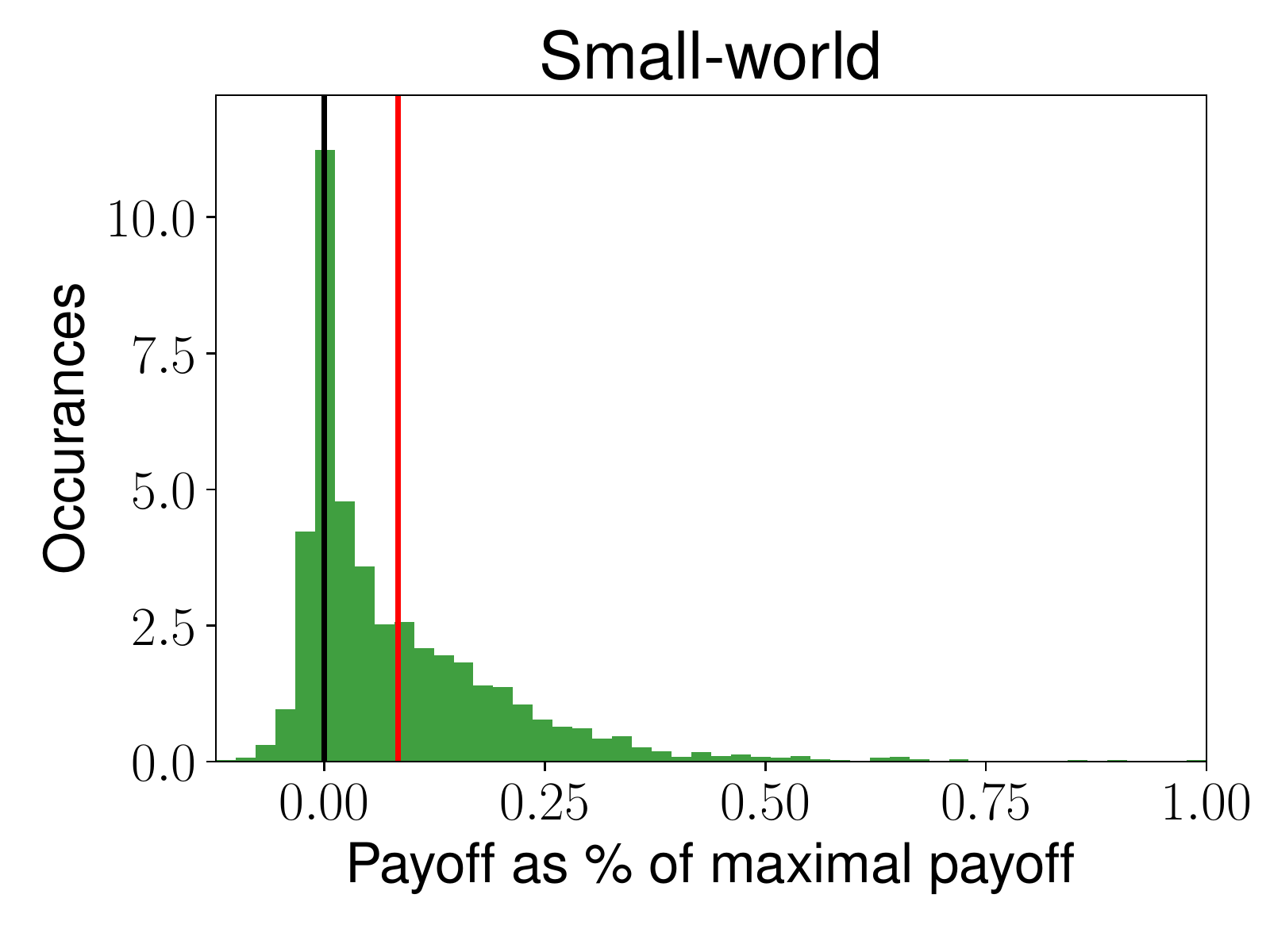}
\includegraphics[width=0.32\linewidth]{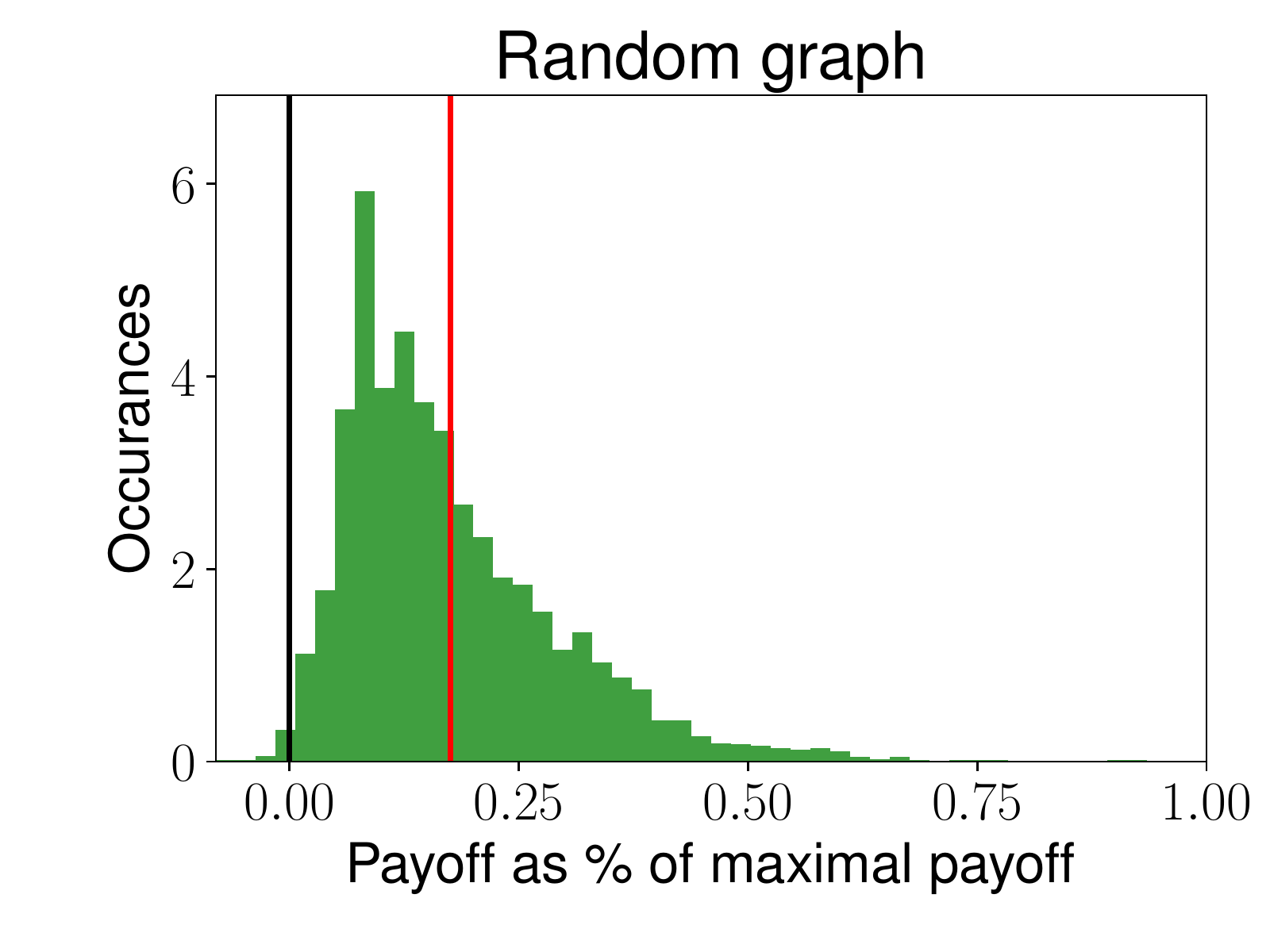}
\caption{The distributions of the evader's payoffs for budget $3$. Values are provided for evader type $\phi = 0.5$ and averaged over the seeker's equilibrium strategies. For each network, the red and black lines denote the average payoff and 0, respectively.}
\label{fig:hists}
\end{figure}

\begin{figure*}[t]
\center
\includegraphics[width=.26\linewidth]{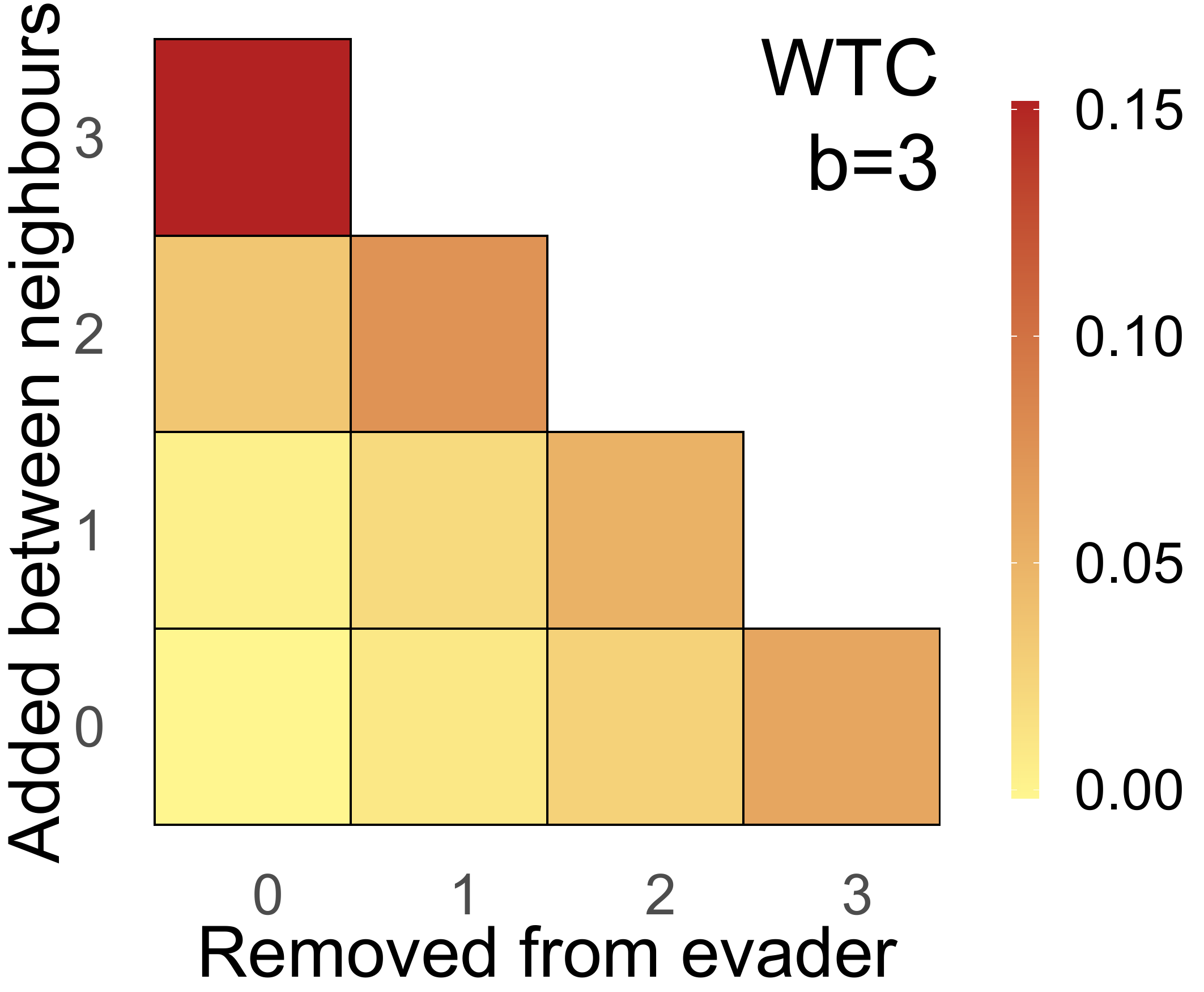}\hfill
\includegraphics[width=.26\linewidth]{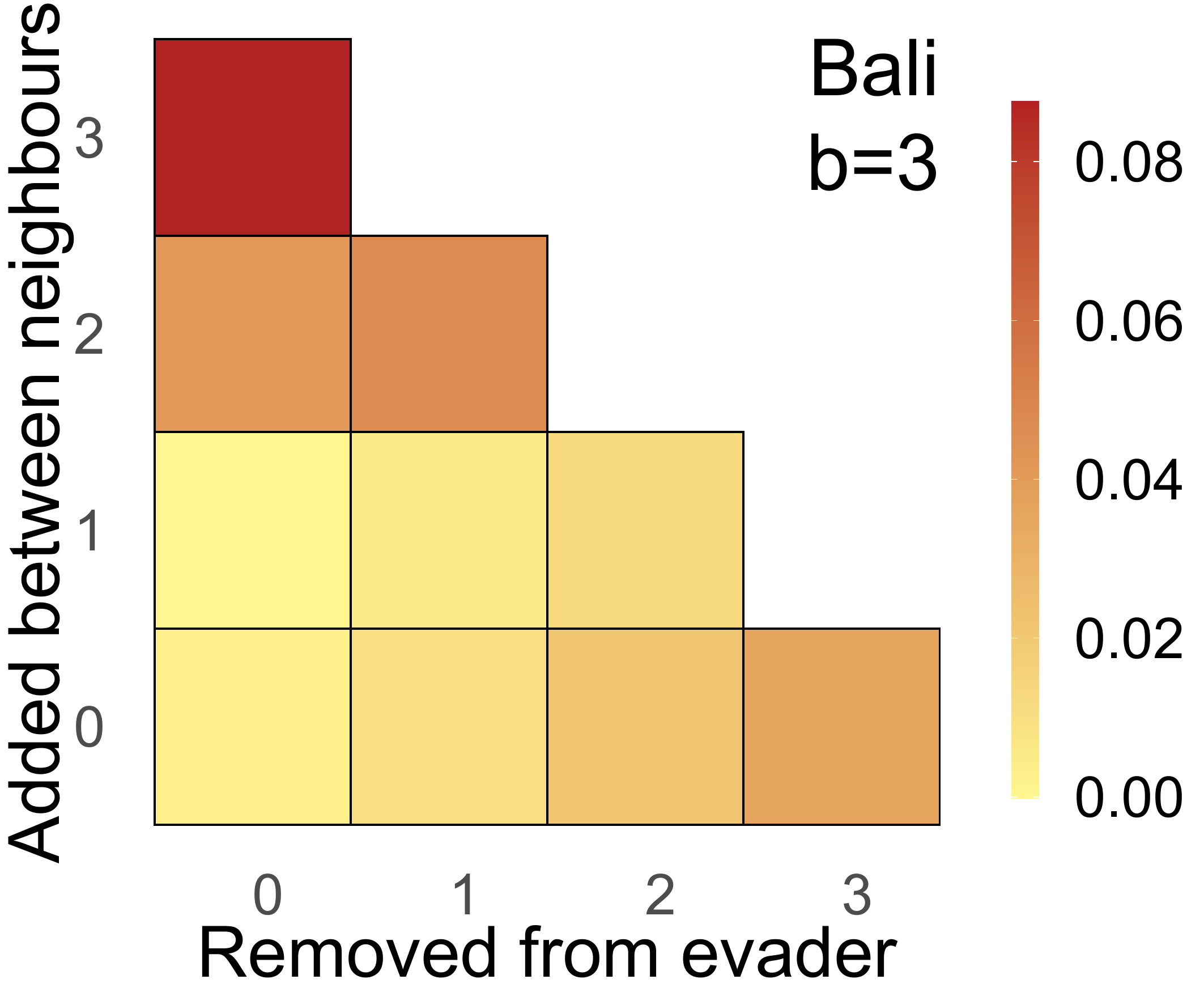}\hfill
\includegraphics[width=.26\linewidth]{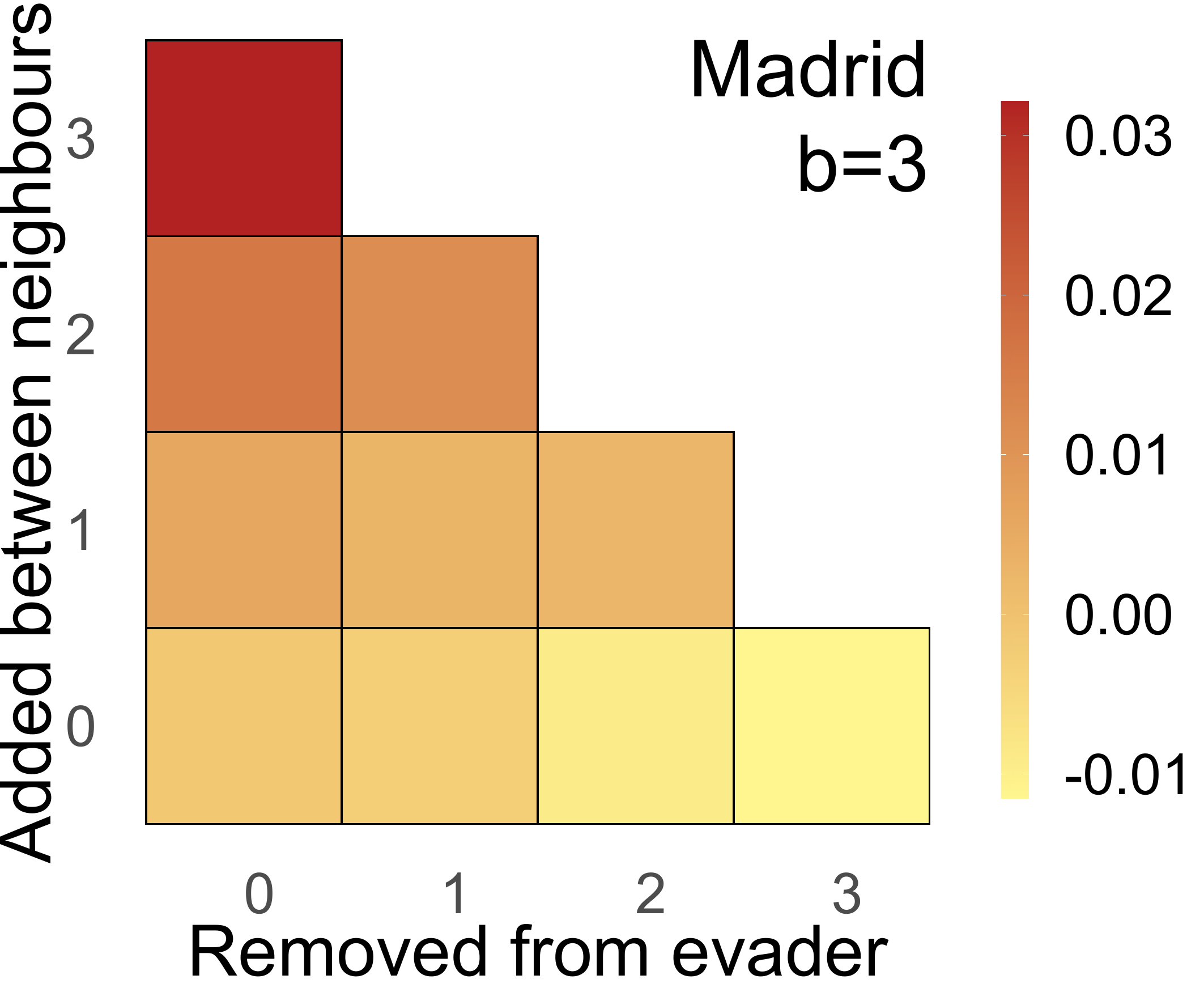}
\includegraphics[width=.26\linewidth]{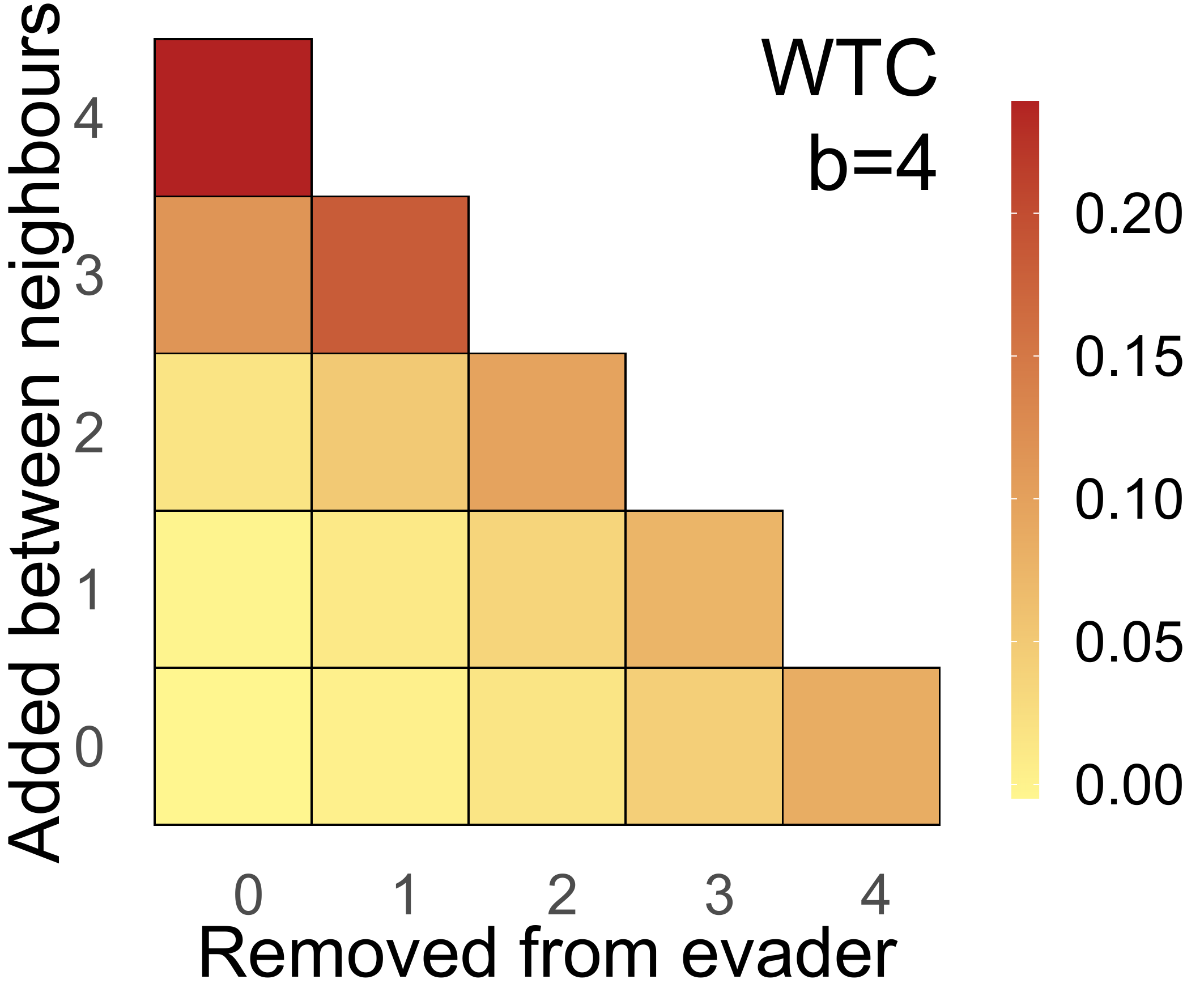}\hfill
\includegraphics[width=.26\linewidth]{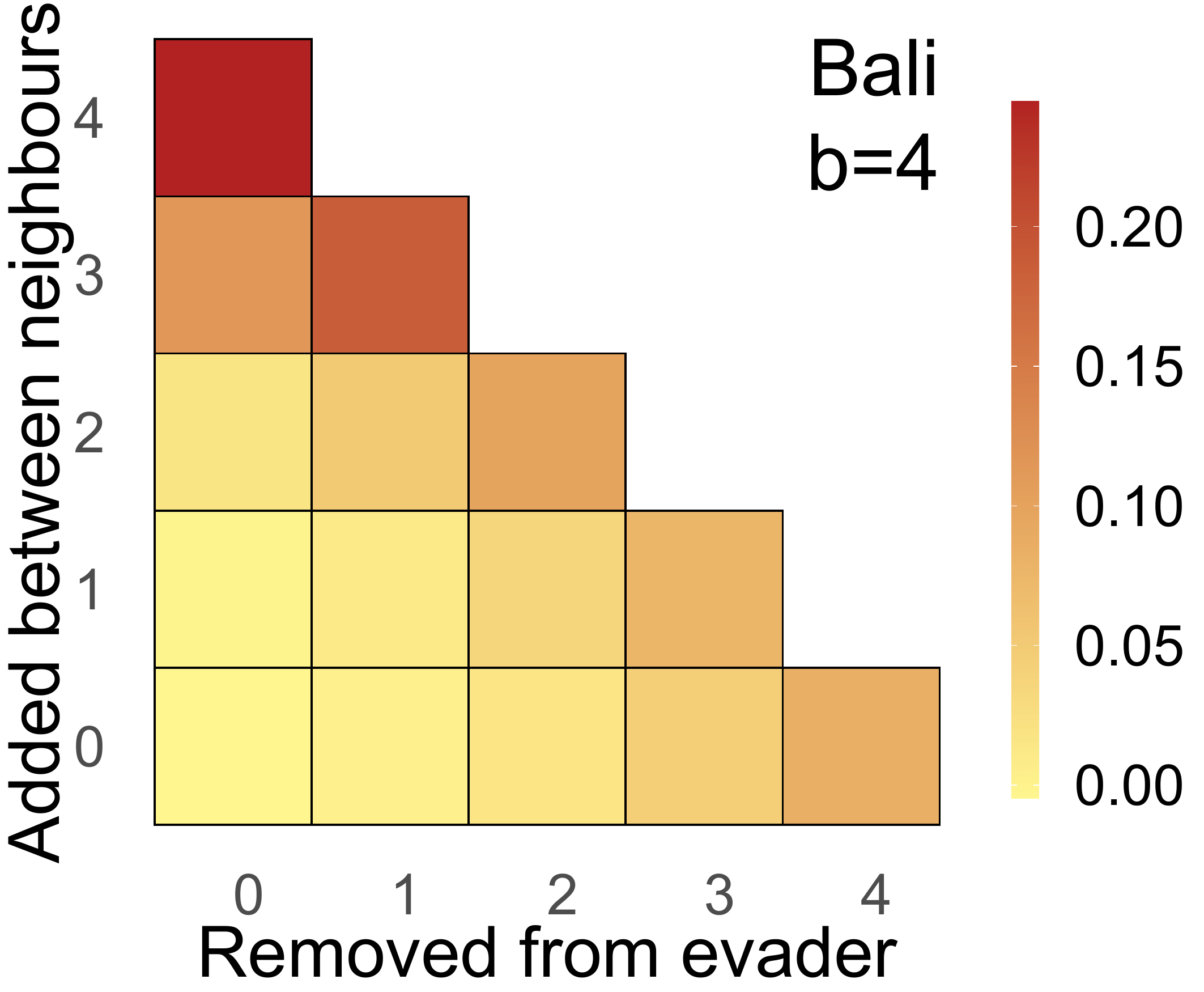}\hfill
\includegraphics[width=.26\linewidth]{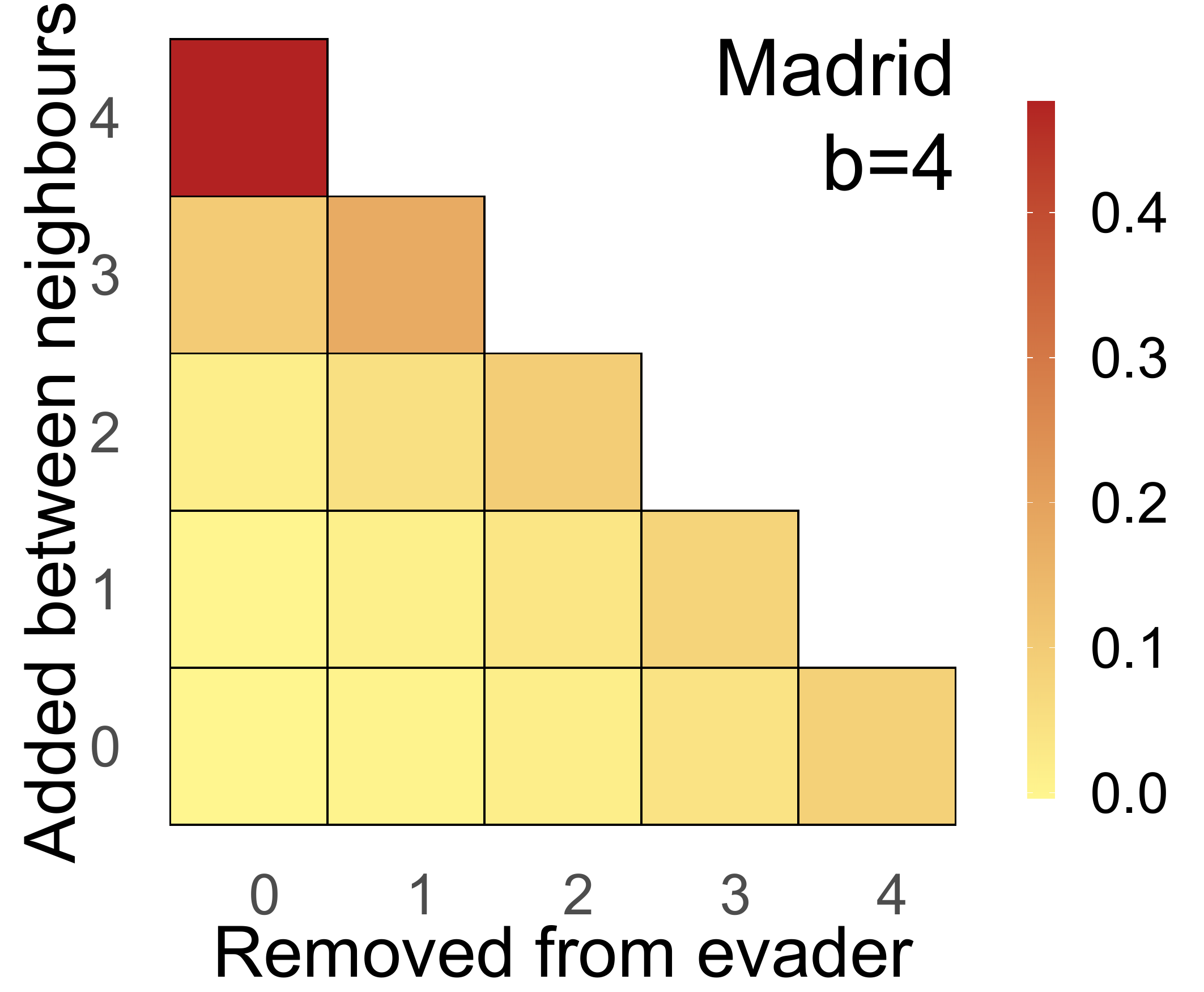}
\includegraphics[width=.26\linewidth]{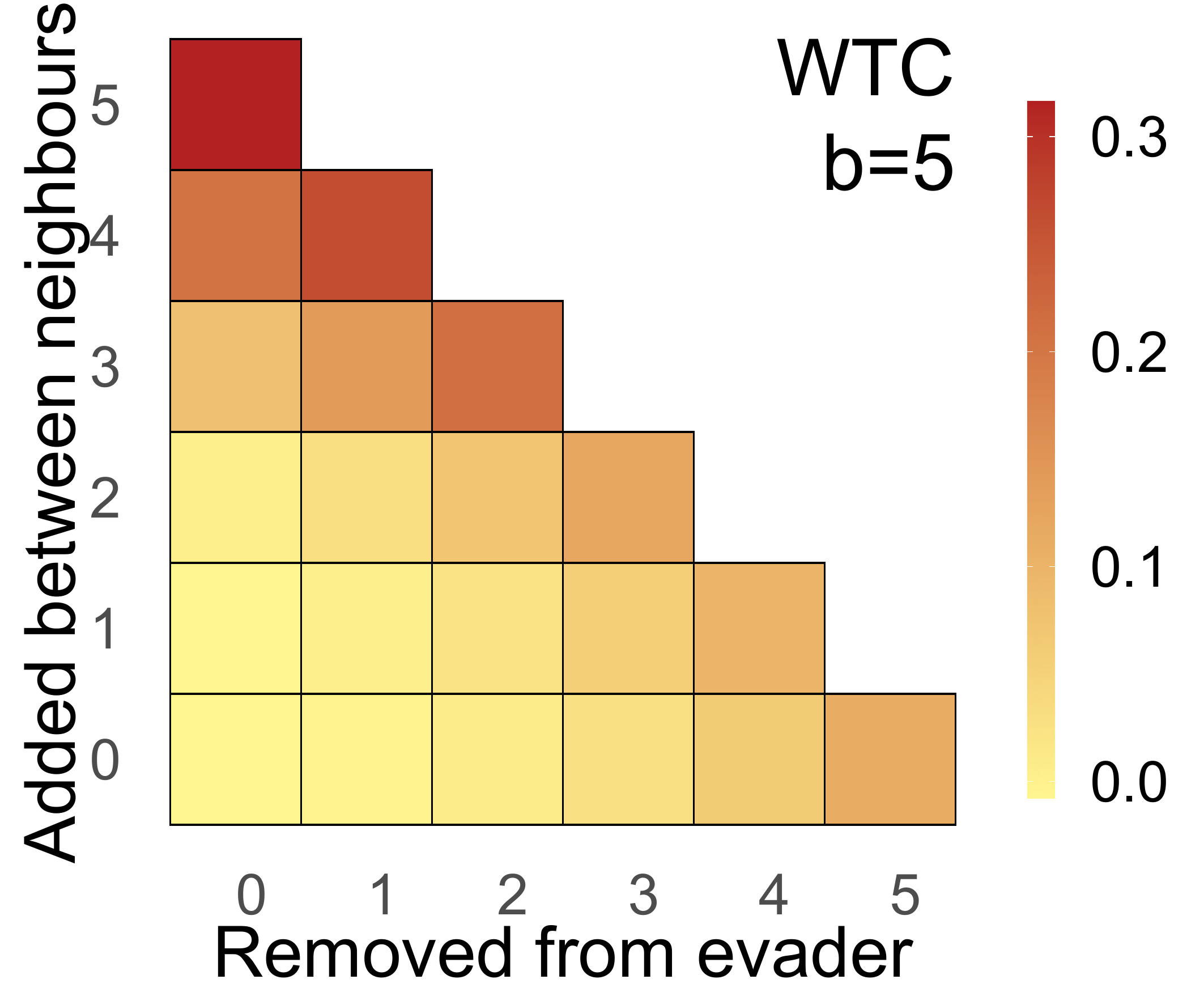}\hfill
\includegraphics[width=.26\linewidth]{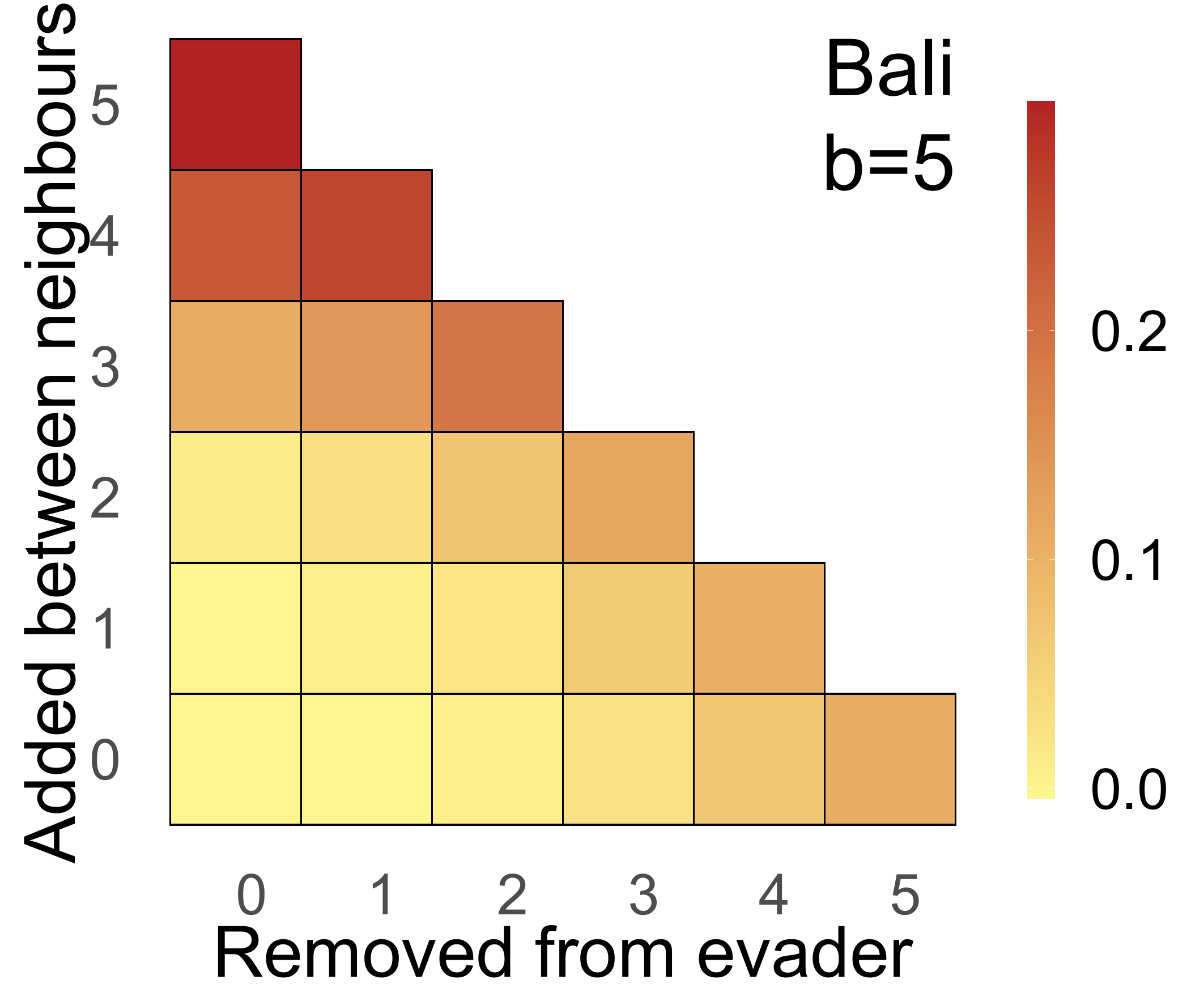}\hfill
\includegraphics[width=.26\linewidth]{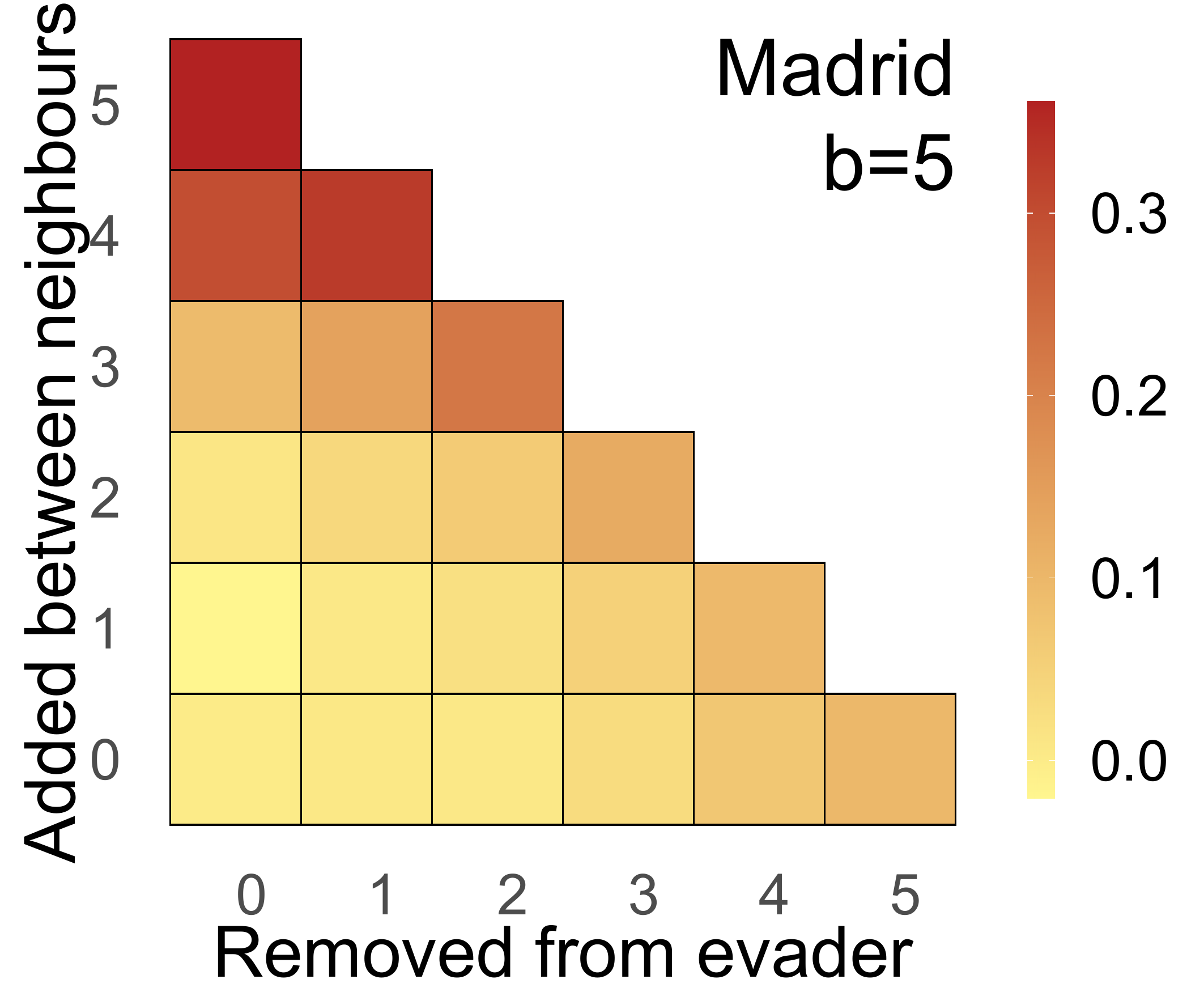}
\caption{
The evader's average payoff given the three terrorist networks (WTC, Bali, and Madrid), and given budgets $3$, $4$, and $5$. The x-axis represents the number of neighbours the evader is disconnected from, while the y-axis represents the number of edges added between the evader's neighbours. The color intensity of each cell represents the evader’s average payoff for given strategy.
%A circle indicates that at least one version of ROAM lies in the corresponding cell.
}
\label{fig:heat:WS:ER}
\end{figure*}

\begin{figure}[t!]
\center
\includegraphics[width=.26\linewidth]{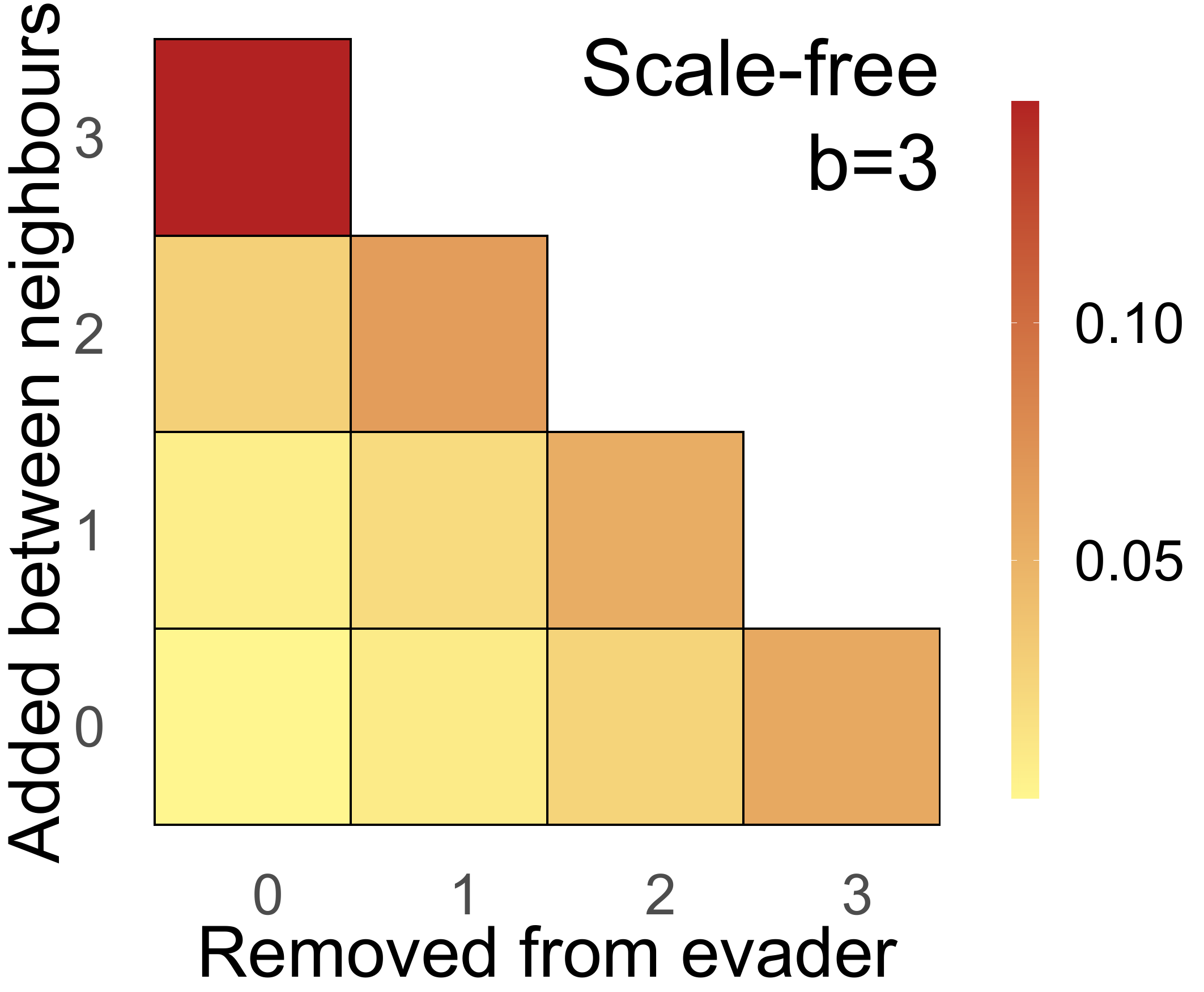}\hfill
\includegraphics[width=.26\linewidth]{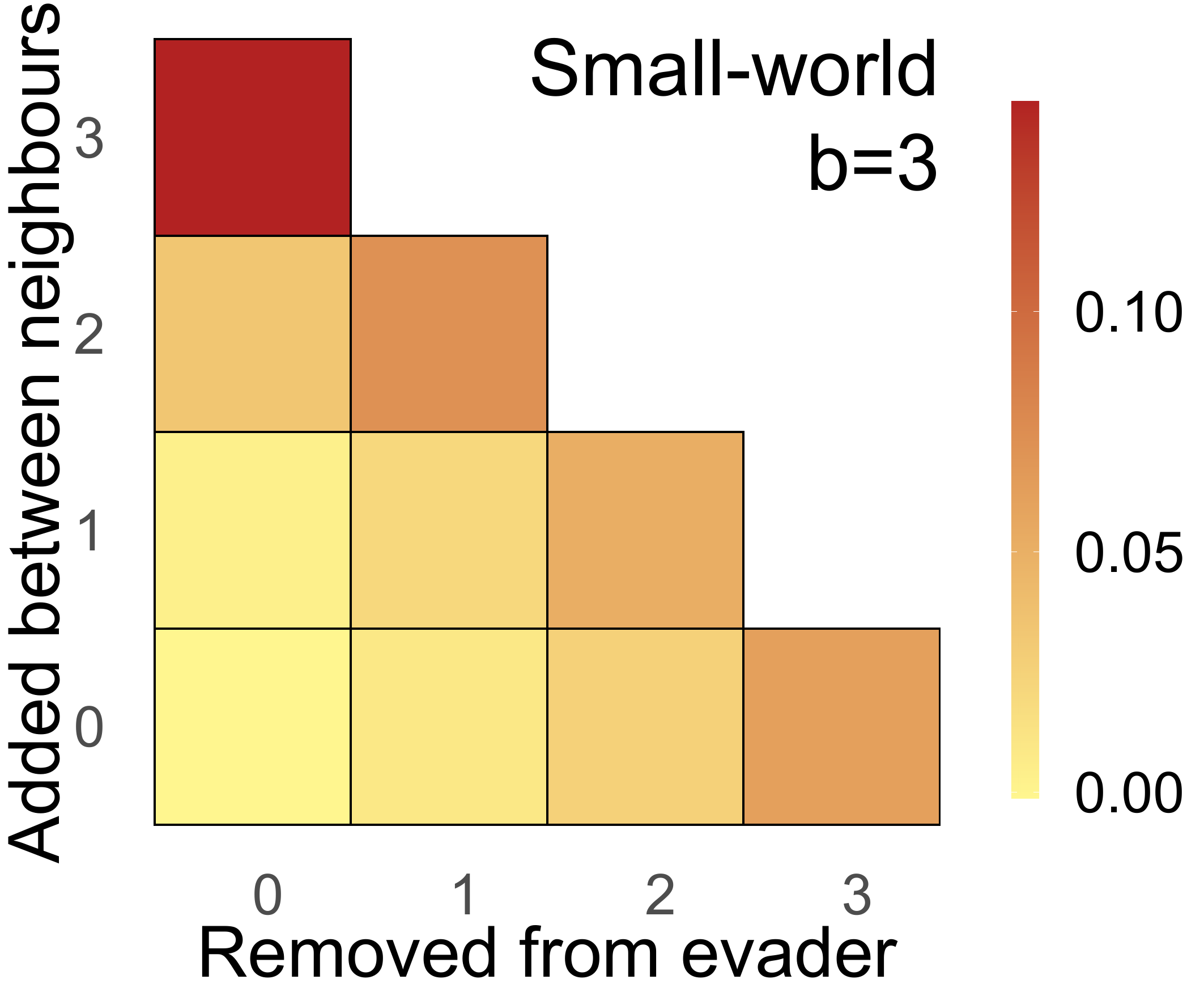}\hfill
\includegraphics[width=.26\linewidth]{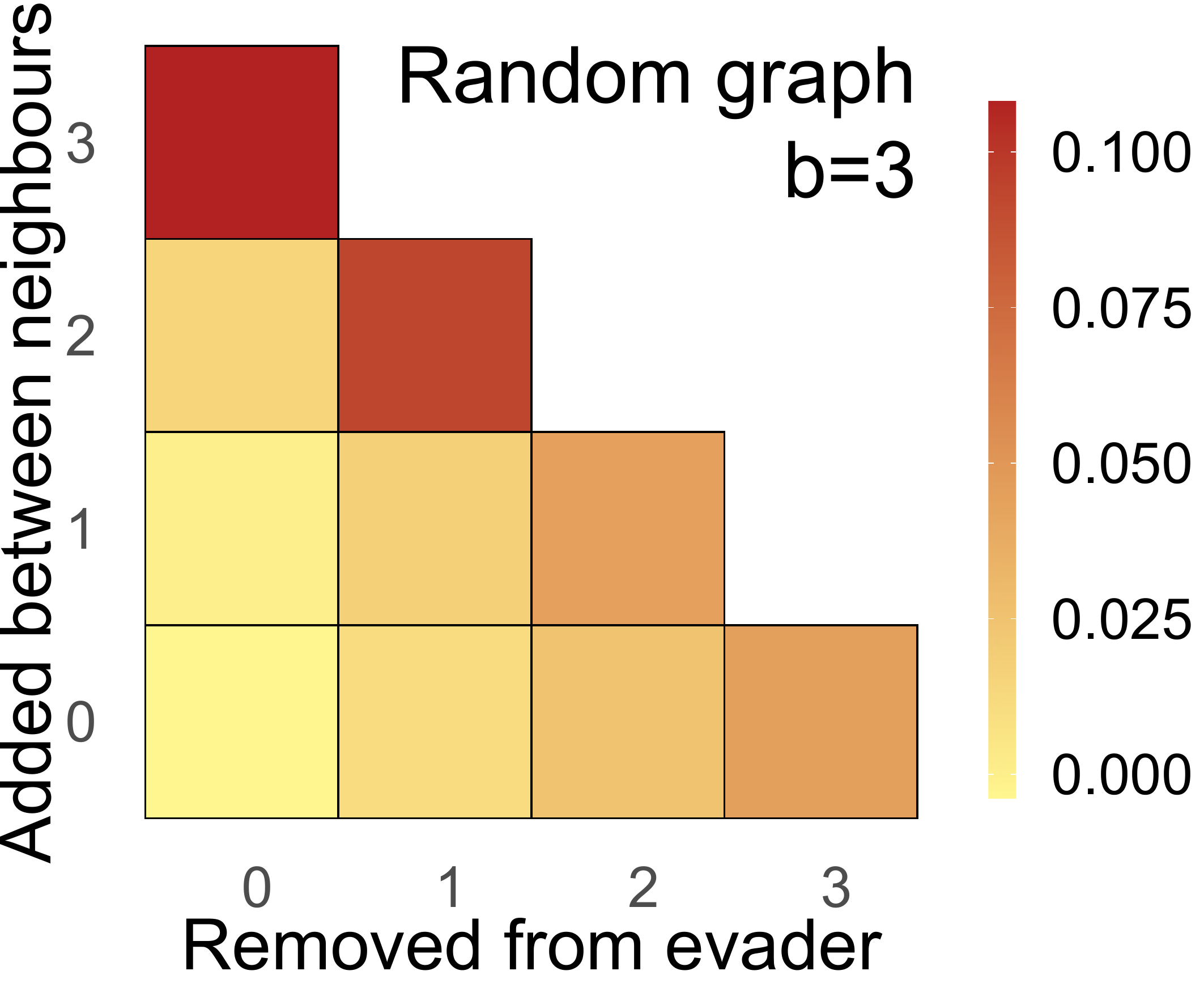}
\includegraphics[width=.26\linewidth]{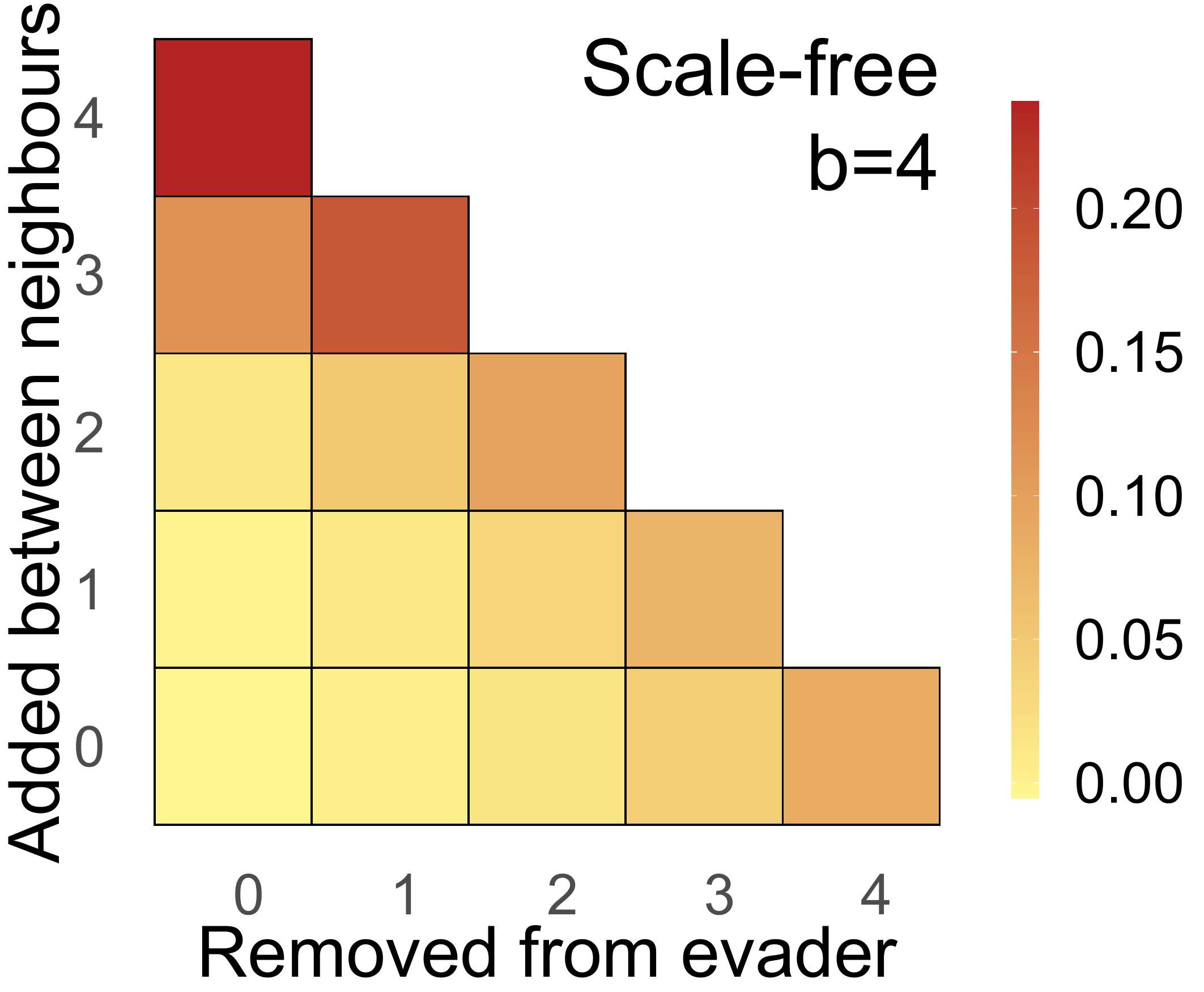}\hfill
\includegraphics[width=.26\linewidth]{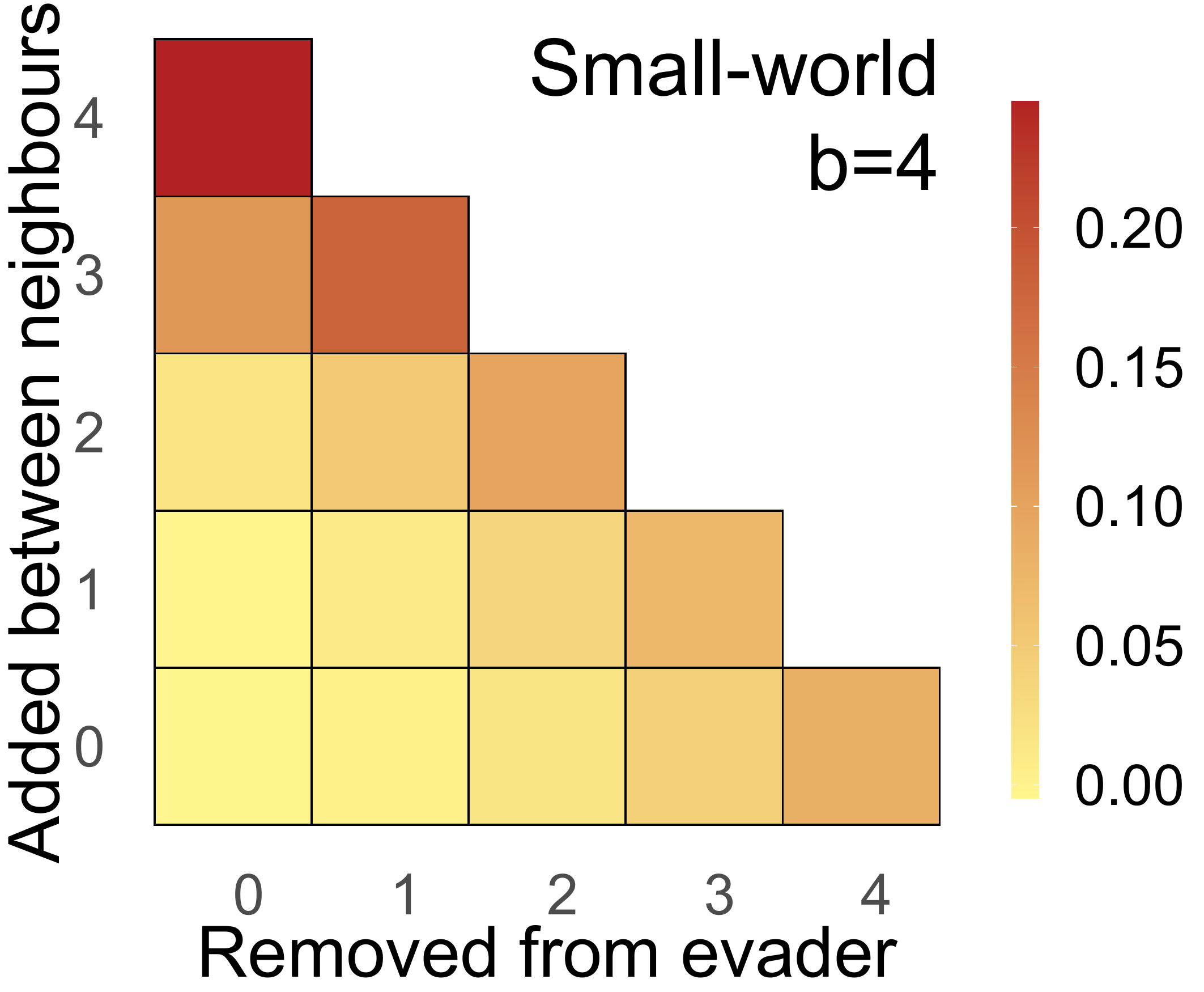}\hfill
\includegraphics[width=.26\linewidth]{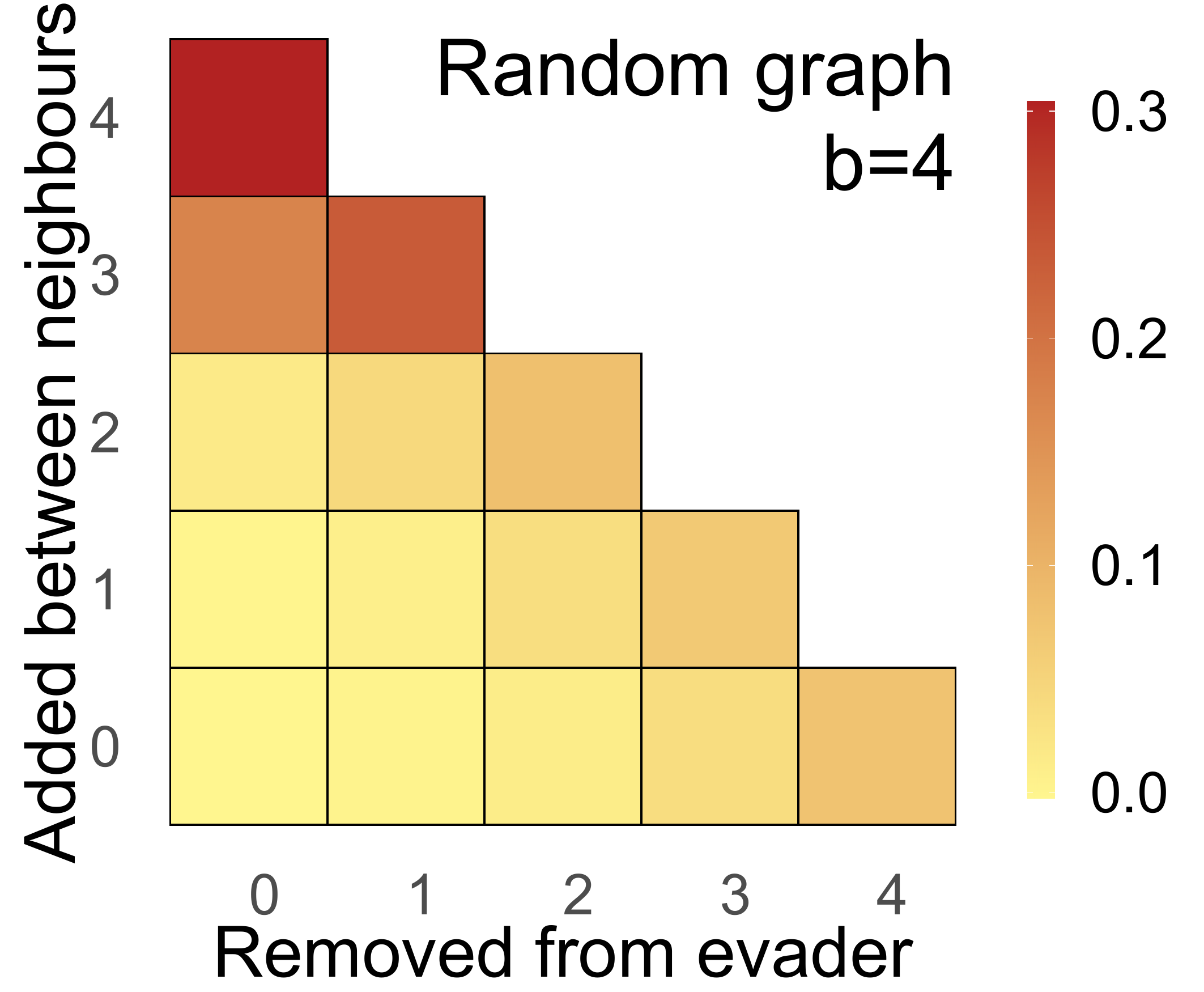}
\includegraphics[width=.26\linewidth]{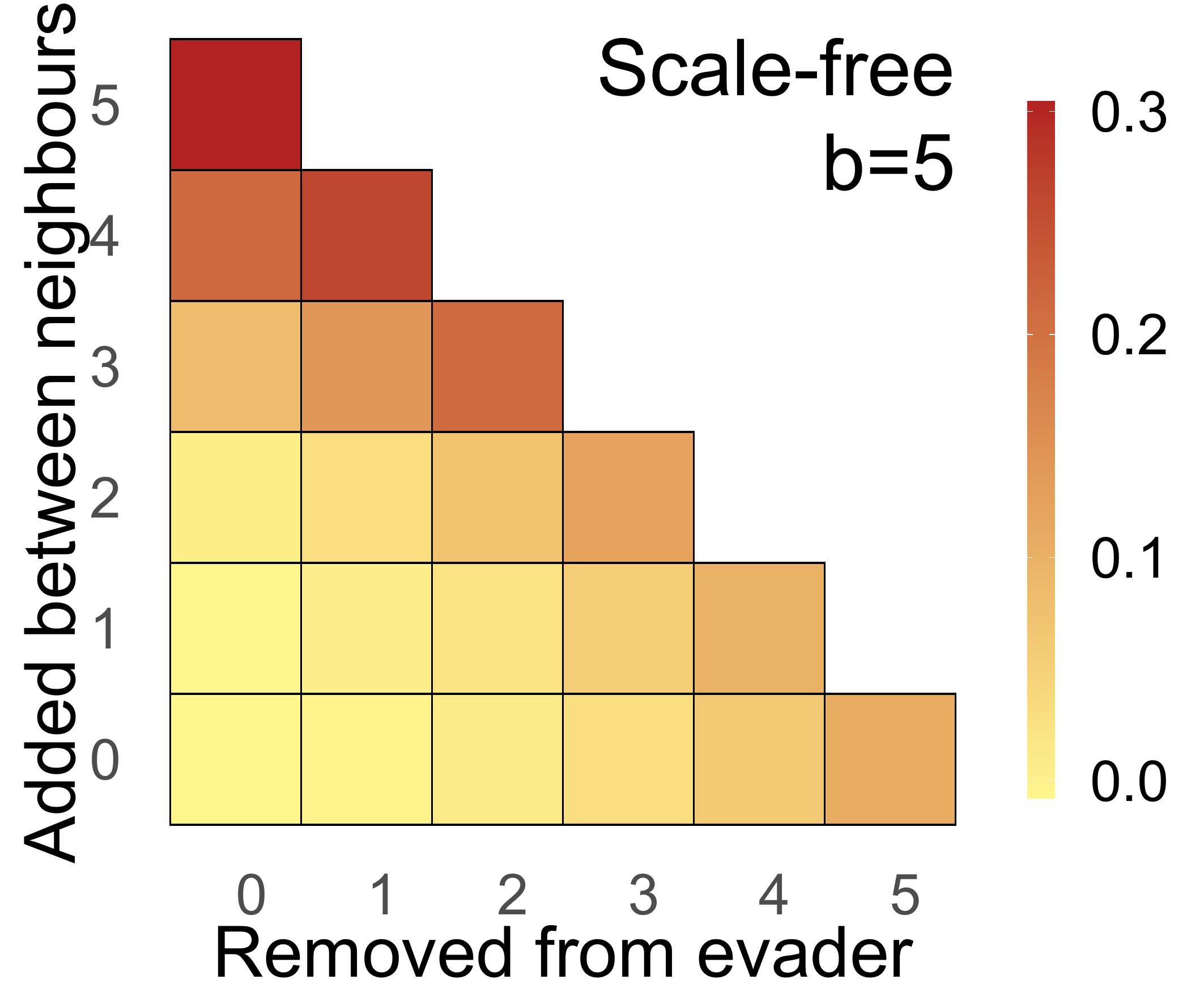}\hfill
\includegraphics[width=.26\linewidth]{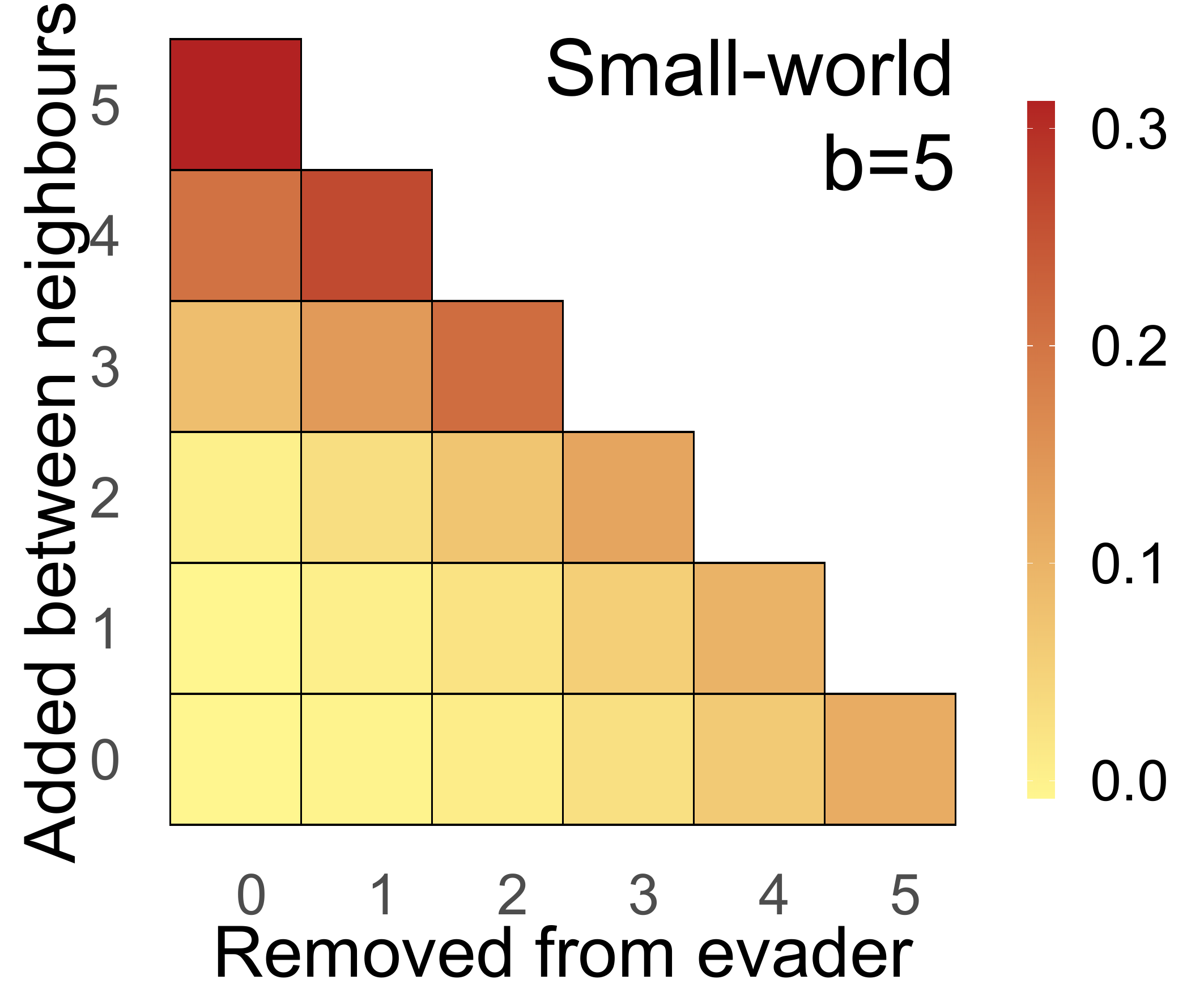}\hfill
\includegraphics[width=.26\linewidth]{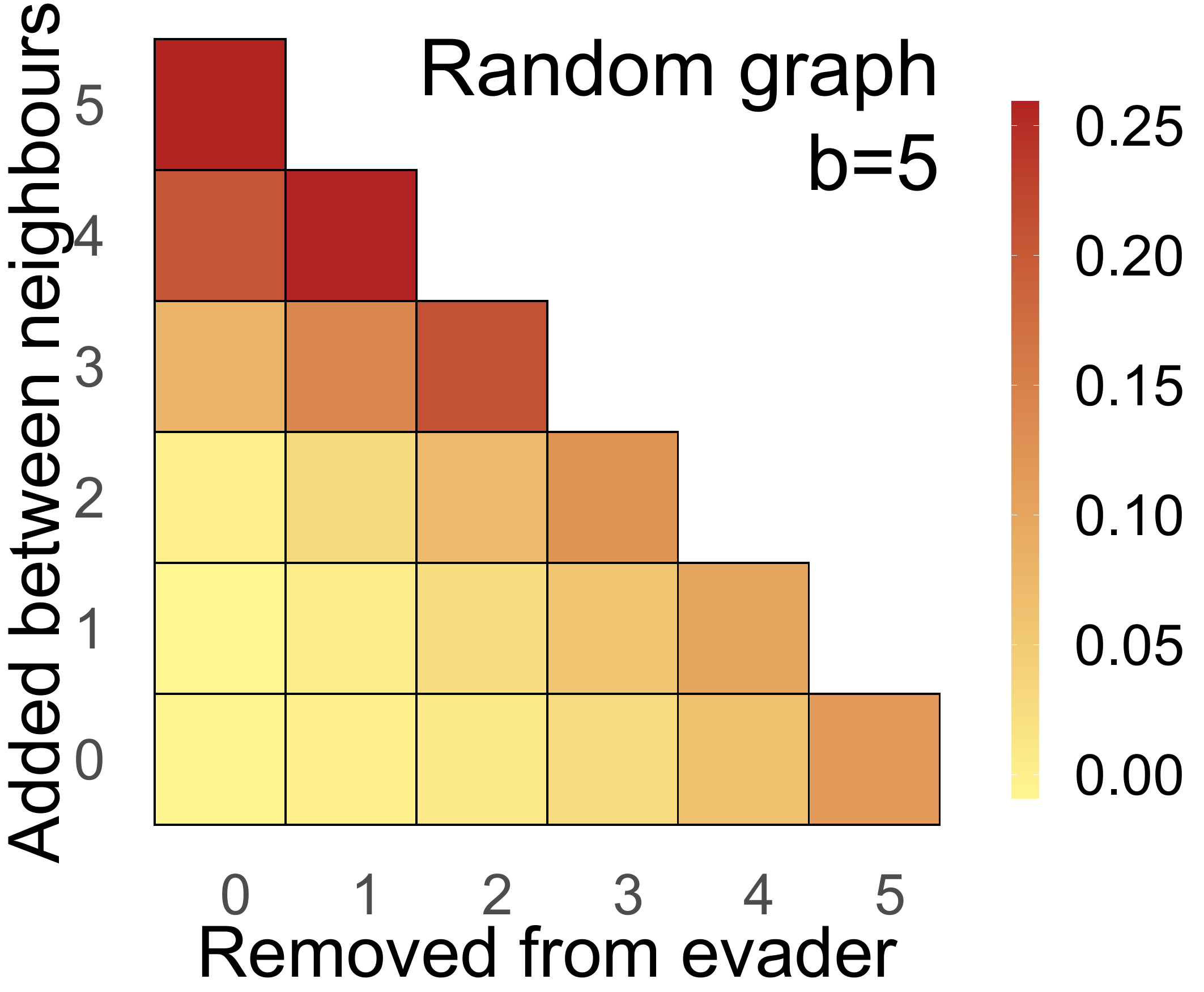}
\caption{
Same as Figure~\ref{fig:heat:WS:ER}, but for scale-free networks, small-world and random-graph networks.
}
\label{fig:heat:wtc:bali}
\end{figure}

Secondly, for any given equilibrium strategy of the evader, the difference in the seeker's payoff between her optimal strategy and other strategies is minimal (less than 1\%). This suggests that, for the zero-sum game, the seeker could, in principle, use any centrality measure to analyse the network, without compromising much efficiency. Conversely,  for any given equilibrium strategy of the seeker, the difference in the evader's payoff between her optimal strategy and other strategies is much more pronounced (more than 100\%, see Figure~\ref{fig:hists}). Hence, \textit{the outcome of the game relies heavily on the evader's choice of strategy, while the seeker's choice of centrality measure has negligible impact}.

Thirdly, the strategies that yield similar payoffs seem to involve rewiring the network in similar ways; see Figures~\ref{fig:heat:WS:ER} and~\ref{fig:heat:wtc:bali}. Interestingly, \textit{the ROAM heuristic of Waniek et al.~\cite{waniek2018hiding} is often among the evader's most rewarding strategies}.

Based on these observations, we next analyze the non-zero-sum version of the game when the evader uses the ROAM heuristic.

\subsection{The Non-Zero-Sum Version}

In this version of the game, we assume that the evader's strategies are instances of the ROAM heuristic. More specifically, the evader's total budget $b$ is used to repeatedly run ROAM. We write ROAM($x$), where $x$ is the number of added between the evader's neighbours. The budget of a single iteration is between $1$ and $\frac{b}{2}$, i.e., there are at least two iterations. The evader repeatedly run ROAM, until the entire budget $b$ is spent. For example, for $b = 10$, we have the following set of evader strategies: $\{$ROAM($1$) repeated $5$ times, ROAM($2$) repeated $3$ times + ROAM($0$), ROAM($3$) repeated twice + ROAM($1$), ROAM($4$) repeated twice$\}$.

We calculate the equilibrium strategy profiles for different networks. For each network, we consider budgets $b \in \{5,10,15,20,25,35\}$, assuming that $b$ is no more than $25\%$ of all edges in the network. This cap is meant to limit the changes in the network characteristics resulting from the evader's actions.

Figure~\ref{fig:mixed} illustrate the mixed strategies played by the seeker in the equilibrium for different networks and evader budgets. For each centrality, Tables~\ref{tab:averages_rg} and~\ref{tab:averages_reallife} present the average probability of being used in different networks.

The equilibrium strategies show, on one hand, which heuristics the evader should use to minimize her centrality while maintaining as much influence as possible. On the other hand, they indicate which centrality the seeker should adopt to have the greatest chance of identifying the evader among the top nodes in the network. Our first key observation in the non-zero-sum game setting is that the choice of the strategy by the seeker has a much greater impact on her payoff than in the zero-sum game. Hence, in what follows, we will focus particularly on the strategies of the seeker, i.e., we will consider which centrality a network analyzer should use when facing a strategic evader.

Regarding the results for the randomly-generated networks, we observe clear, robust patterns, suggesting that it is possible to identify some combination(s) of centrality measures that can be used against the evader. In particular:
\begin{itemize}
\item \emph{Scale-free networks:} degree centrality is used almost exclusively. Due to the power-law distribution of nodes' degrees in scale-free networks, the ``hubs'' have extremely high degree, and the evader is most certainly one of them. As such, even with a large budget, any attempts to reduce the evader's position in the degree-based ranking have limited impact.
\item \emph{Small-world networks:} eigenvector centrality consistently proves to be most difficult to manipulate, it is played by the seeker in almost every small-world network.
\item \emph{Random graph networks:} For low values of the evader's budget, eigenvector centrality is the most effective. However, for larger budgets, it is often replaced by closeness centrality. This shift occurs when budget reaches about $15$, regardless of the network size.
\end{itemize}

\begin{figure*}[th]
\centering
\begin{subfigure}[b]{.32\linewidth}
    \center
 	\includegraphics[width=\linewidth]{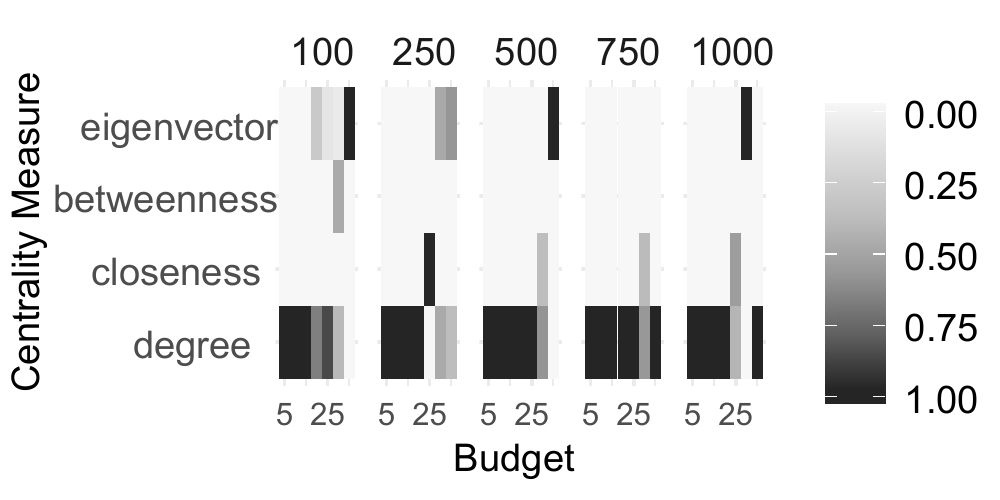}
    \caption{Scale-free}
\end{subfigure}
\begin{subfigure}[b]{.32\linewidth}
    \center
    \includegraphics[width=\linewidth]{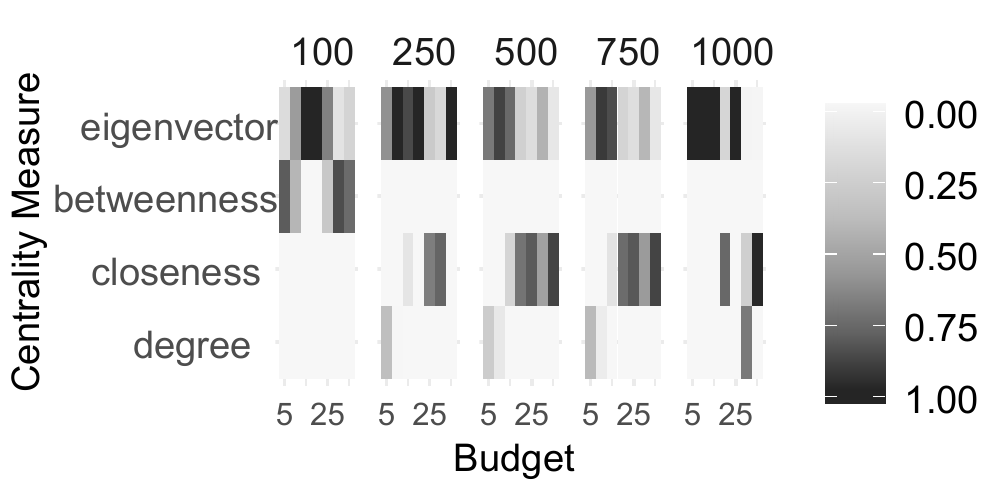}
    \caption{Random graphs}
\end{subfigure}
\begin{subfigure}[b]{.32\linewidth}
    \center
    \includegraphics[width=\linewidth]{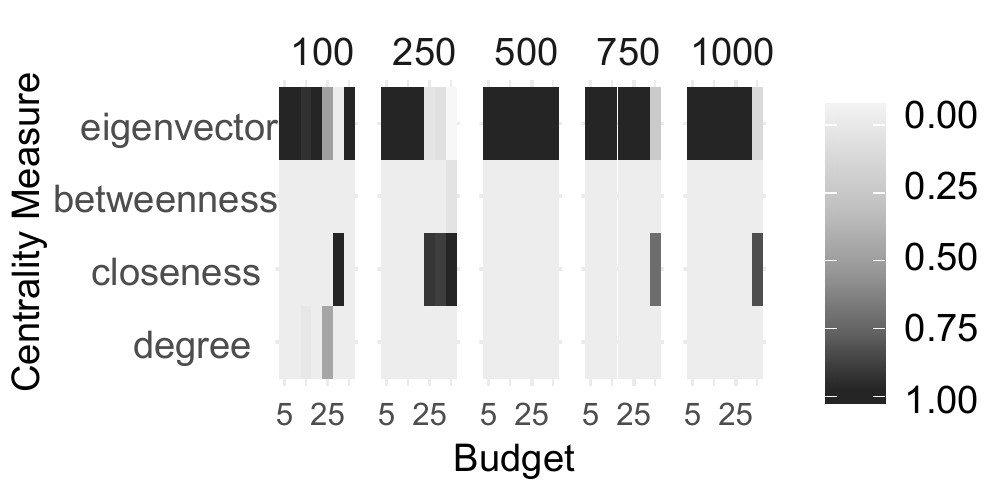}
    \caption{Small-world}
\end{subfigure}
\begin{subfigure}[b]{.32\linewidth}
    \center
    \includegraphics[width=\linewidth]{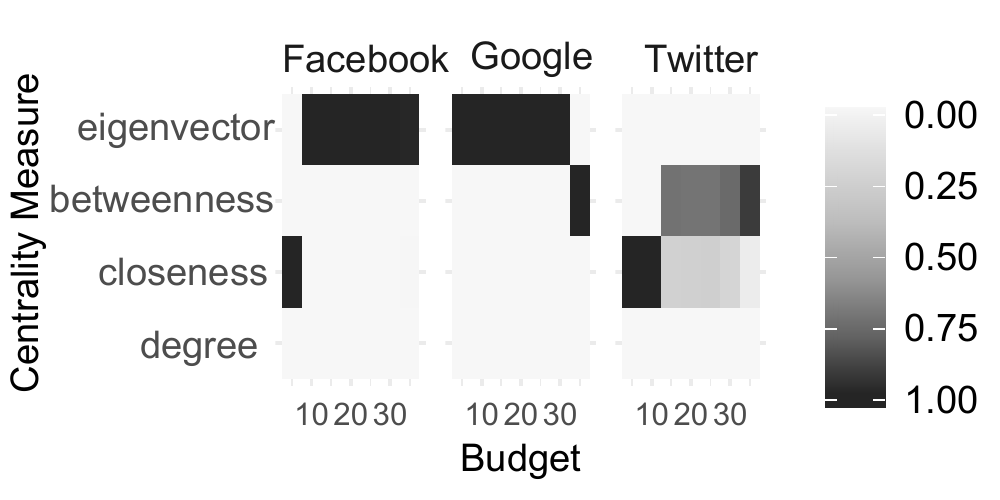}
    \caption{Social media}
\end{subfigure}
\begin{subfigure}[b]{.32\linewidth}
    \center
    \includegraphics[width=\linewidth]{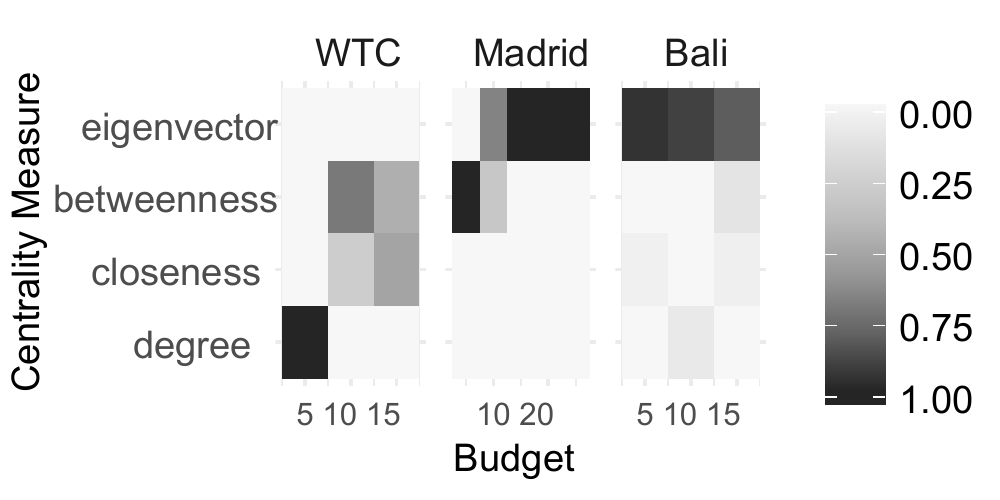}
    \caption{Terrorist networks}
\end{subfigure}
\caption{
The seeker's equilibrium strategies given the evader types $\{0.2,0.4,0.6,0.8\}$, in (a) \emph{scale-free}, (b) \emph{random graph} and (c) \emph{small-world} networks with $100$, $250$, $500$, $750$ and $1000$ nodes, as well as in (d) social media and (e) terrorist networks. Results are presented for $d=15$, and the independent cascade influence model. A darker color indicates that the corresponding centrality measure has a greater weight in the seeker's mixed strategy.
}
\label{fig:mixed}
\end{figure*}

\begin{table}[t]
\centering
\begin{tabular}{lcccc}  
\toprule
Network & $c_{betw}$ & $c_{clos}$ & $c_{degr}$ & $c_{eig}$ \\
\midrule
\emph{Scale-free} & 0 & 0.04 & 0.94 & 0.04 \\
\emph{Random graphs} & 0.05 & 0.08 & 0.25 & 0.62 \\
\emph{Small-world} & 0 & 0 & 0.06 & 0.94 \\
\bottomrule
\end{tabular}
\caption{The average probability of using each centrality given randomly-generated networks.}
\label{tab:averages_rg}
\end{table}

\begin{table}[t]
\centering
\begin{tabular}{lcccc}  
\toprule
Network & $c_{betw}$ & $c_{clos}$ & $c_{degr}$ & $c_{eig}$ \\
\midrule
WTC & 0.04 & 0.03 & 0.03 & 0.89 \\
Bali & 0.39 & 0.27 & 0.33 & 0 \\
Madrid & 0.27 & 0 & 0 & 0.73 \\
\midrule
Overall Terrorist & 0.23 & 0.10 & 0.12 & 0.54 \\
\midrule
Facebook & 0 & 0.14 & 0 & 0.86 \\
Google+ & 0 & 0.14 & 0 & 0.86 \\
Twitter & 0 & 0.56 & 0.44 & 0 \\
\midrule
Overall Social & 0 & 0.28 & 0.15 & 0.57 \\
\bottomrule
\end{tabular}
\caption{The average probability of using each centrality given different real-life networks.}
\label{tab:averages_reallife}
\end{table}

Regarding the results for the real-life networks, we also find regularities. Overall, for the networks with lower average clustering coefficient and lower density (Madrid and WTC attacks, Facebook, Google+), eigenvector centrality seems to be played most often. Furthermore, degree centrality is never played against the evader in larger networks. In more detail:

\begin{itemize}
\item \emph{Covert organizations}: for the WTC 9/11 attack and the Madrid train attack networks, eigenvector centrality is played almost exclusively. On the other hand, for the Bali attack network, degree and betweenness centralities are chosen. This last network, in addition to being the smallest, consists of two subnetworks connected by one node---Samudra---the leader of the terrorist organization. This atypical topology of the network may be responsible for the difference. Moreover, the average clustering coefficient and the density for the Bali network are much greater than for the other networks.
\item \emph{Social media}: eigenvector centrality is the most frequent choice for Facebook and Google+ networks, but for the Twitter network it is replaced by closeness and betweenness. This could be due to the former networks having a lower density and average clustering coefficient than the last one, making them more similar to small-world networks.
\end{itemize}

The above analysis of equilibrium strategies, both for real-life and randomly-generated networks, allows us to derive a number of policy recommendations:

\begin{itemize}
\item Eigenvector centrality should be used by the seeker in networks exhibiting small-world properties. This finding is supported by the results for both randomly generated small-world networks and real-life social media networks.
\item Degree centrality should be used by the seeker in scale-free networks, as evident by the results for Barabasi-Albert networks. However, since those networks exhibit some small-world properties, eigenvector centrality can be considered as a second choice.
\item For networks that resemble random graphs, eigenvector centrality proves to be useful, at least against evaders whose budget is small. As for larger budgets, closeness centrality yields superior results.
\item For two of the three terrorist networks under consideration, eigenvector centrality dominates the alternatives, highlighting its potential benefits when facing covert networks. 
\end{itemize}

In general, eigenvector centrality seems to be a reliable choice for a variety of network types. Although for some networks it is the second best choice, generally it outperforms other measures, and seems to be more resilient against strategic manipulation.

\section{Conclusions}

We investigated the problem of concealing the importance of an individual in a social network, where both the evader, i.e., the person who wishes to hide, and the seeker, i.e., the party analyzing the network, act strategically. We focused on settings where the evader cannot rewire edges between complete strangers, but instead can only modify connections involving her neighbours in the networks. We showed that even in this simplified setting, the problem of finding an optimal way to hide from the most fundamental centrality measures is NP-complete. In light of these hardness results, we analyzed a number of instances of the game under both the zero-sum and the non-zero-sum payoffs; this highlighted some potential policy implications for network analyzers in the face of a strategic evader.

For future work, we intend to study this setting more rigorously, e.g., by analyzing the case in which multiple evaders are acting simultaneously, and more broadly, e.g., by considering a wider range of centrality measures available to the seeker. Another interesting follow-up of this study is to analyze the problem of hiding from link-prediction algorithms under the assumption that both the evader and the seeker act strategically.

\section*{Acknowledgments}

Tomasz Michalak was supported by the Polish National Science Centre (grant 2016/23/B/ST6/03599).
Yevgeniy Vorobeychik was supported by the National Science Foundation (IIS-1903207, IIS-1905558) and Army Research Office (MURI W911NF1810208).
Kai Zhou was supported by PolyU (UGC) Internal Fund (1-BE3U).
For an earlier version of this work, Marcin Waniek was supported by the Polish National Science Centre (grant 2015/17/N/ST6/03686).

\clearpage

\bibliographystyle{abbrv} 
\bibliography{bibliography-strategic-evasion}

\end{document}